\documentclass[12pt]{article}
\usepackage{amssymb}
\usepackage{latexsym}
\usepackage{graphicx}
\usepackage{amsmath}
\usepackage{hyperref,amsthm}
\usepackage{epsfig}
\usepackage{booktabs}
\usepackage{amsmath}
\usepackage{appendix}
\usepackage{enumerate}

\numberwithin{equation}{section}
\normalsize \makeatletter
\newtheorem{thm}{Theorem}[section]
\newtheorem{lemma}[thm]{Lemma}

\newtheorem{corollary}[thm]{Corollary}
\newtheorem{definition}[thm]{Definition}

\newtheorem{rem}[thm]{Remark}

\newcommand{\be}{\begin{equation}}
\newcommand{\ee}{\end{equation}}
\newcommand{\bea}{\begin{eqnarray*}}
\newcommand{\eea}{\end{eqnarray*}}
\newcommand{\mR}{\mathbb{R}}

\newcommand{\mB}{\mathcal{B}}

 \DeclareMathOperator*{\argmin}{arg\,min}

\newcommand{\la}{\langle}
\newcommand{\ra}{\rangle}

\newenvironment{keyword}{\smallskip\noindent{\bf Keywords.}
                          \hskip\labelsep}{}

\begin{document}
\date{}
\title{ \textbf {Restricted $q$-Isometry Properties Adapted to
Frames for Nonconvex $l_q$-Analysis \footnote{This work is supported
by Key project of NSF of China under number 11531013 and NSF of China under number 11171299.
\newline J. Lin is with the Department of Mathematics, City University of Hong Kong, Kowloon, Hong Kong, China (Email: jhlin5@hotmail.com).
\newline
S. Li is with the Department of Mathematics, Zhejiang University, Hangzhou, China. He is the  corresponding author (E-mail: songli@zju.edu.cn).}}
\author{Junhong Lin and Song Li }
} \maketitle
\begin{abstract}
This paper discusses reconstruction of signals from few measurements
in the situation that signals are sparse or approximately sparse in
terms of a general frame via the $l_q$-analysis optimization with
$0<q\leq 1$. We first introduce a notion of restricted $q$-isometry
property ($q$-RIP) adapted to a dictionary, which is a natural
extension of the standard $q$-RIP, and establish a generalized
$q$-RIP condition for approximate reconstruction of signals via the
$l_q$-analysis optimization. We then determine how many random,
Gaussian measurements are needed for the condition to hold with high
probability. The resulting sufficient condition is met by fewer
measurements for smaller $q$ than when $q=1$.

The introduced generalized $q$-RIP is also useful in compressed data
separation. In compressed data separation, one considers the problem
of reconstruction of signals' distinct subcomponents, which are
(approximately) sparse in morphologically different
dictionaries, from few measurements. With the notion of generalized
$q$-RIP, we show that under an usual assumption that the
dictionaries satisfy a mutual coherence condition, the $l_q$ split
analysis with $0<q\leq1 $ can approximately reconstruct the distinct
components from fewer random Gaussian measurements with small $q$
than when $q=1$.

\begin{keyword}
Compressed sensing, Restricted isometry property, Frames,
$l_q$-analysis, Sparse recovery, Data separation.
\end{keyword}
\end{abstract}

\section{Introduction}
\subsection{Background} Recovery of signals which are (approximately) sparse in
terms of a dictionary from few measurements is one of the major
subjects in compressed sensing. Suppose that we observe data from
the model
$$y=Af,$$ where $A\in\mR^{m\times n}$ (with $m<n$) is a known measurement
matrix. Our goal is to reconstruct the unknown signal $f$ based on
$y$ and $A$.

In standard compressed sensing \cite{CT06,CRT06b,D06b}, one assumes
that $f$ is sparse in the standard coordinate basis. A vector $v$ is
$s$ sparse if it has at most $s$ nonzero entries. If the measurement
matrix $A$ satisfies a restricted isometry property (RIP) condition
$\delta_{cs}\leq C$ (see e.g. \cite{CRT06b,CZ12} and the references
therein), one can recover a sparse signal $f$ by solving an
$l_1$-minimization problem
$$\min\limits_{\tilde{f}\in{\mR^n}}\|\tilde{f}\|_1\quad\mbox{subject
to}\quad A\tilde{f}=y.\eqno{(L_{1})}$$ Recall that a matrix $A$ is
said to satisfy the RIP \cite{CT06} of order $s$ if there is some
$\delta\in[0, 1)$ such that, for all $x$ with $\|x\|_0\leq s$, we
have
$$(1-\delta)\|x\|_2^2\leq\|Ax\|_2^2\leq(1+\delta)\|x\|_2^2.$$
The infimum of all possible $\delta$ satisfying the above inequality, denoted as $\delta_s$, is the so-called RIP constant of order $s$.
Many types of random measurement matrices such as Gaussian matrices
or Sub-Gaussian matrices have the RIP constant $\delta_s\leq\delta$
with overwhelming probability provided that $ m\geq
C\delta^{-2}s\log (n/s)$ \cite{CT06,BDDW08,MPT08,RV08}. Based on its
RIP guarantees, with high probability, $(L_1)$ can recover every $s$
sparse vector from $O(s\log (n/s))$ random measurements.

One alternative way of finding the unknown signal proposed in the
literature is to solve
$$\min_{\tilde{f}\in{\mR^n}}\|\tilde{f}\|_q\quad\mbox{subject
to}\quad A\tilde{f}=y. \eqno{(L_{q})}$$ Reconstructing sparse
signals via $(L_q)$ with $0<q<1$ has been considered in a series of
papers (see e.g. \cite{CS08,SCY08,FL09,DDFG10,GN07} and the
references therein) and some of the virtues are highlighted
recently. The $l_q$-strategy offers an advantage in that it requires
fewer measurements in numerical experiments \cite{C07}, with random and nonrandom Fourier measurements.
 Chartrand and Staneva \cite{CS08} showed that if
$A$ is an $m\times
 n$ Gaussian matrix, every $s$ sparse vector $f$ can be exactly
 recovered by solving $(L_{q})$ with high probability provided
 $$m\geq C_1(q)s + qC_2(q)s \log(n/s),$$
 where $C_1(q)$ and $C_2(q)$ are bounded and given explicitly there. The
 dependence of $m$ on the number $n$ of columns vanishes for
 $q\rightarrow0.$ In their proof, they used a restricted $q$-isometry property, namely
 $$(1-\delta)\|v\|_2^q\leq\|Av\|_q^q\leq(1+\delta)\|v\|_2^q$$ for
 all $s$ sparse vectors $v\in\mR^n$ and $0<q\leq 1.$

In this paper, the signal is assumed to be (approximately) sparse in
terms of a frame $D$, i.e., $D^*f$ is (approximately) sparse. Some
examples in practice are Gabor frames \cite{FS98} in radar and
sonar, curvelet frames \cite{CD04} and undecimated wavelet
transforms \cite{M08,CDOS12} in image processing, etc. Recall that
the columns of $D\in\mR^{n\times d}$ ($n\leq d$) form a frame for
$\mR^n$ with frame bounds $0<\mathcal{L}\leq\mathcal{U}<\infty$ if
\be\label{Frames}\forall f\in\mR^n,\quad \mathcal{L}\|f\|_2^2\leq\|D^*f\|_2^2\leq \mathcal{U}\|f\|_2^2.\ee
If $\mathcal{U}=\mathcal{L}$, then $D$ is a tight frame for $\mR^n$.
One way of recovery such signals is via the following $l_q$ analysis (see e.g. \cite{EMR07} and the reference therein)
with $0<q\leq 1$:
$$\hat{f}=
\argmin_{\tilde{f}\in{\mR^n}}\|D^*\tilde{f}\|_q\quad\mbox{subject
to}\quad A\tilde{f}=y. \eqno{(P_{q})}$$ We remark that $(P_{q})$ may
have more than one minimizer, and our results of this paper hold for
any solution of $(P_{q})$. Here, for simplicity of statements, we
assume that $(P_{q})$ has a unique minimizer.
 Letting $D$ be a tight
frame, Cand\`{e}s et al. \cite{CENR11} showed that the solution
$\hat{f}$ of ($P_1$) satisfies $$\|\hat{f}-f\|_2\leq
C_0\frac{\|D^*f-(D^*f)_{[s]}\|_1}{\sqrt{s}},$$ provided that $A$
satisfies an $D$-RIP condition. Here we denote $x_{[s]}$ to be the
vector consisting of the $s$ largest coefficients of $x\in\mR^d$ in
magnitude:
$$x_{[s]}=\argmin_{\|\tilde{x}\|_0\leq s}\|x-\tilde{x}\|_2.$$ Recall
that a measurement matrix $A$ is said to obey the restricted
isometry property adapted to $D$ (abbreviated as $D$-RIP)
\cite{CENR11} of order $s$ if there exists some $\delta \in (0,1)$ such that
\be\label{StandardDRIP}(1-\delta)\|Dv\|_2^2\leq\|ADv\|_2^2\leq(1+\delta)\|Dv\|_2^2\ee
holds for all $s$ sparse vectors $v\in\mR^d$. The $D$-RIP constant of order $s$, denoted as $\delta_s$, is the infimum of all possible $\delta$ satisfying the above inequality.
Note that the $D$-RIP
is a natural extension of the standard RIP. Under the assumption
that $A$ satisfies an $D$-RIP condition, for general $0<q\leq 1$,
\cite{ACP11,LL11} provided results on recovery of signals which are
compressible in terms of a tight frame $D$ via $(P_q).$ Liu et al.
\cite{LML12} considered the problem of recovering signals which are
compressible in a general frame $D$ via dual frame based
$l_1$-analysis model. Nam et al. \cite{NDEG13} proposed a new signal
model called cosparse analysis model with corresponding
reconstruction methods. In a recent paper, Rauhut and Kabanava \cite{KR13} provided
both uniform and nonuniform recovery guarantees from Gaussian random
measurements, which requires $O(s\log (d/s))$ measurements, for
cosparse signals based on $(P_1)$ when $D$ is a frame.

The $D$-RIP is a special case of a more general definition given in
\cite{BD09,LD08}. Until now, nearly all good constructions of $D$-RIP measurement matrices uses randomness. For any choice of $D\in\mR^{n\times d}$, if $A$ is
populated with independent and identically distributed (i.i.d.)
random entries from a Gaussian or Sub-Gaussian distribution, then
with high probability, $A$ will satisfy the $D$-RIP of order $s$ as
long as $m = O(s \log(d/s))$ \cite{CENR11,BD09,LD08}. In fact, given
any matrix $A$ satisfying the traditional RIP, by applying a random
sign matrix one obtains a matrix satisfying the $D$-RIP \cite{KW11}.
Based on its $D$-RIP guarantees, the aforementioned results show
that $(P_q)$ with $0<q\leq 1$ can guarantee approximately recovery
from $O(s\log(d/s))$ measurements for Sub-Gaussian matrices.


\subsection{Main contribution}
In this paper, we further develop theoretical results on $l_q$
analysis for approximate recovery of signals, that are approximately
sparse with respect to a general frame $D$. One of our main results shows that $(P_q)$ can
approximately recover the unknown signal with high probability from
fewer measurements with small $q$ than that were needed in the
aforementioned results. Concretely, we have the following result.
\begin{thm}\label{thm2}
Suppose that we observe data from the model $y=Af.$  Let
$D\in\mR^{n\times d}$ be a frame with frame bounds
$0<\mathcal{L}\leq\mathcal{U}<\infty$. Let $A$ be an $m\times n$
matrix whose entries are i.i.d. random distributed normally with
mean zero and variance $\sigma^2.$ Then there exist constants
$C_1(q)$ and $C_2(q)$ such that whenever $0<q\leq 1$ and
$$m\geq
C_1(q)\kappa^{\frac{q}{2-q}}s+qC_2(q)\kappa^{\frac{2q}{2-q}}s\log(d/s),
  \quad \kappa=\mathcal{U}/\mathcal{L},$$ with
probability exceeding $1-1/\binom{d}{s}$, any solution $\hat{f}$ of
$(P_q)$ satisfies
  $$\|\hat{f}-f\|_2\leq
  C\frac{\|D^*f-(D^*f)_{[s]}\|_q}{s^{1/q-1/2}}.$$
\end{thm}
\begin{rem}
  \begin{enumerate}
  \item $C_1(q)$ and $C_2(q)$ are bounded positive numbers which will be explicitly given in the proof
  of the theorem.
    \item The dependence of $m$ on the number $d$ and the condition number $\kappa$ of $D$
     vanishes for $q\rightarrow0.$ As a result, the required
     measurements become $Cs$ when $q$ is small, which are fewer than that were needed in the previous
     results.
  \item Using the proof techniques developed in \cite{SL12}, one can
  improve the success probability.
  \end{enumerate}
\end{rem}
The proof of this result is based on a notion of $(D,q)$-RIP and general $(D,q)$-RIP recovery result.
It is a natural extension of the standard $q$-RIP in \cite{CS08}
(and for $q=1$, in \cite{D06b}):
\begin{definition}[($D$,q)-RIP] Let $D$ be an $n\times d$ matrix. A measurement matrix $A$ is said to obey the
restricted $q$-isometry property adapted to $D$ (abbreviated as
($D$,q)-RIP) of order $s$ with constant $\delta\in[0,1)$ if
\be\label{DRIP}(1-\delta)\|Dv\|_2^q\leq\|ADv\|_q^q\leq(1+\delta)\|Dv\|_2^q\ee
holds for all $s$ sparse vectors $v\in\mR^d$. The $(D,q)$-RIP
constant $\delta_s$ is defined as the smallest number $\delta$ such
that (\ref{DRIP}) holds for all $s$ sparse vectors $v\in\mR^d$.
\end{definition}
In section 2, we first establish approximate recovery results for $l_q$ analysis
with the assumption that the measurement matrix $A$ satisfies an
 $(D^{\dag},q)$-RIP condition. Here $D^{\dag}=(DD^*)^{-1}D$ is the
canonical dual frame of $D$. Subsequently, we prove Theorem \ref{thm2} by showing how many
random Gaussian measurements are sufficient for the condition to
hold with high probability. The resulting sufficient condition is
met by fewer measurements for smaller $q$ than when $q=1$.

Our approach ($P_q$) with $q=1$ is slightly different with the
$l_1$-analysis approach considered in \cite{LML12,F14}, i.e.,
\be\label{L1Analysis} \argmin_{\tilde{f}}\|(D^{\dag})^*\tilde{f}\|_1
\quad\mbox{s.t.}\quad A\tilde{f}=y.\ee Using $D$ (instead of
$D^{\dag}$) as analysis operator is preferable in some certain
circumstances, e.g., when $D$ is known while $D^{\dag}$ is hard to
be known or computed in high dimension sparse recovery, or when $D$
is of special structure which has fast computation algorithm while
$D^{\dag}$ is not (see e.g. \cite{DH02}). We also note that it is
hard to verify the ($D,q$)-RIP (or $D$-RIP) for a deterministic
measurement matrix, and for certain random measurement matrice,
verifying the ($D^{\dag},q$)-RIP (or $D^{\dag}$-RIP) is almost the
same as verifying the ($D,q$)-RIP (or $D$-RIP).

The proof techniques in this paper may shed some lights on improving
the existing $D$-RIP recovery results. Our proof for the
$(D^{\dag},q)$-RIP guaranteeing results in Section 2.4 shows that a
suitable $(D^{\dag},q)$-RIP condition implies the $l_q$ null space
property of order $s$ relative to $D$ ($D$-NSP$_q$) of the
measurement matrix. Recall that the matrix $A$ is said to satisfy
the $D$-NSP$_q$ of order $s$ \cite{GN03,CDD09,S11,ACP11,F14} if
there exists a constant $\theta$ with $\theta\in(0,1)$ such that for
all $h\in \ker A$ and for all sets $T\subset \{1,\cdots d\}$ with
cardinality $|T|\leq s,$
$$\|D_T^*h\|_q^q \leq \theta \|D^*_{T^c}h\|_q^q.$$
Here, $D_T$ is the submatrix of $D$ formed from the columns of $D$ indexed by $T$.
The smallest value of the constant $\theta$ in the above is referred to
as the $D$-NSP$_q$ constant.
The importance of the null space property is that it is the necessary
and sufficient condition under which $l_q$ recovery is exact for
$s$-sparse signals for the case $D=I$ (see e.g. \cite{GN07,CDD09}). By developing a tighter
relationship between the $D$-RIP constant and the $D$-NSP constant, one can improve the RIP condition for the exact sparse
recovery (see e.g. \cite{CZ12} for the case $D=I$ and \cite{LL11} for the case of tight frames).
For general frames case, it would be interesting to pursue a tighter
relationship between the $D^{\dag}$-RIP constant and the $D$-NSP constant, and then
establish an $D^{\dag}$-RIP recovery result for ($P_1$) using the
approach of this paper. However, this is beyond the scope of this
paper.

\subsection{Compressed
data separation} Numeral examples show that signals of interest
might be classified as multimodal data, i.e., being composed of
distinct subcomponents. One common task is to separate such data
into appropriate single components for further analysis (e.g.
\cite{ESQD05,CCS10,DK13,CDOS12}). In \cite{LLS13TIT}, the authors
considered data separation from few measurements, and showed that
the two distinct subcomponents, assumed to be approximately $s_1$
and $s_2$ sparse in two dictionaries $D_1\in\mR^{n\times d_1}$ and
$D_2\in\mR^{n\times d_2}$ respectively, can be approximately
reconstructed by solving the $l_1$ split analysis, provided that the
measurement matrix satisfies an $D$-RIP condition and the two
dictionaries satisfy a mutual coherence (between the different
dictionaries) condition. Based on the $D$-RIP analysis, under a
mutual coherence condition between the two dictionaries, the $l_1$
split analysis can approximately reconstruct the distinct components
from $O((s_1+s_2)\log \frac{d_1+d_2}{s_1+s_2})$ random sub-Gaussian
measurements. Refer to \cite{LLS13TIT} and the references therein
for more details on compressed data separation.

Our second contribution of this paper is to establish further theoretical
recovery results for compressed data separation via $l_q$ split
analysis from random Gaussian measurements. With the $(D,q)$-RIP
introduced in this paper, and  under an usual assumption that the
two dictionaries satisfy a mutual coherence condition, we show that
the $l_q$ split analysis with $0<q\leq1 $ can approximately
reconstruct the distinct components from fewer random Gaussian
measurements with small $q$ than that are needed in previous
results.
Recall that the mutual coherence between two dictionaries
\cite{LLS13TIT} is defined as follow.
\begin{definition} Let $D_1=(d_{1i})_{1\leq i\leq d_1}$ and $D_2=(d_{2j})_{1\leq j\leq d_2}$.
The mutual coherence between $D_1$ and $D_2$ is defined as
$$\mu_1=\mu_1(D_1;D_2)=\max_{i,j}|\la d_{1i},d_{2j}\ra|.$$\end{definition}
We have the following result, whose proof will be given in Section 3 by applying a general theorem
for compressed data separation where unknown signals are composed of $\iota$ ($\iota\geq2$) components that are sparse in terms of $r$ tight frames $D_1,D_2,\cdots,D_r$.
\begin{thm}\label{thm3}
Suppose that we observe data from the model $y=A(f_1+f_2).$  Let
$D_1\in\mR^{n\times d_1}$ and $D_2\in\mR^{n\times d_2}$ be two
arbitrary tight frames for $\mR^n$ with frame bound $1$,
respectively. Let $A$ be an $m\times n$ matrix whose entries are
i.i.d. random distributed normally with mean zero and variance
$\sigma^2.$ Fix positive integers $s_1$ and $s_2$. Assume that the
mutual coherence $\mu_1$ between $D_1$ and $D_2$ satisfies
$$\mu_1(s_1+s_2)\left(\lceil(2^{3q/2}5)^{\frac{2}{2-q}}\rceil+1\right)\left(\frac{1}{8\cdot
5^{2/q}}+1\right)<1.$$
 Then there exist constants $C_1(q)$ and $C_2(q)$ such
that whenever $0<q\leq 1$ and
$$m\geq
C_1(q)(s_1+s_2)+qC_2(q)(s_1+s_2)\log
\left(\frac{d_1+d_2}{s_1+s_2}\right),$$ with probability exceeding
$1-1/\binom{d_1+d_2}{s_1+s_2}$, any solution $(\hat{f}_1,\hat{f}_2)$
to the $l_q$ Split analysis \be\label{SAA} (\hat{f}_1,\hat{f}_2)=
\argmin\limits_{\tilde{f}_1,\tilde{f}_2\in{\mR^n}}\|D_1^*\tilde{f}_1\|_q^q+\|D_2^*\tilde{f}_2\|_q^q
\quad\mbox{s.t.}\quad A(\tilde{f}_1+\tilde{f}_2)=y,\ee obeys
\be\label{R2}\|\hat{f_1}-f_1\|_2+\|\hat{f_2}-f_2\|_2\leq
C_1\frac{\left(\|D_1^*f_1-(D_1^*f_1)_{[s_1]}\|_q^q+\|D_2^*f_2-(D_2^*f_2)_{[s_2]}\|_q^q\right)^{1/q}}{(s_1+s_2)^{1/q-1/2}}.\ee
\end{thm}
\begin{rem}
As $q$ becomes smaller, weaker mutual coherence condition and fewer measurements are needed to guarantee approximate recovery. In particular,
  letting $q \to 0$, the mutual coherence condition and the required measurements become $2\mu_1(s_1 + s_2)<1$ and $m= O(s_1+s_2),$ respectively.
\end{rem}
As we will see  in Section 3, Theorem \ref{thm2} can be generalized to the cases where signals $f$ are composed of general $\iota \in \mathbb{Z}^{+}$ distinct components.
To the best of our knowledge, our results may be the first of this kind for a general $\iota$.
For simplicity, we have restricted to the tight frames case. Note
that similar as Theorem \ref{thm2}, our result can be extend to the
non-tight frames case.

The proof is similar to that of Theorem \ref{thm2}. Under the
assumptions that the measurement matrix satisfies a generalized
$q$-RIP condition, and that the dictionaries satisfy a mutual
coherence condition, we first prove that the $l_q$ split analysis
with $0<q\leq1 $ can approximately reconstruct the distinct
components. Subsequently we determine how many random, Gaussian
measurements are sufficient for the generalized $q$-RIP condition to
hold with high probability. The resulting sufficient condition is
met by fewer measurements for smaller $q$ than when $q=1$. Such a
proof is given in Section 3.

\subsection{Notation}
For a vector $v\in\mR^d$, $\|v\|_0$ is the number of nonzero entries
of $v$. For any $q\in (0,\infty)$, denote
$\|u\|_q=\left(\sum_{j=1}^{d}|u_j|^q\right)^{1/q}$ and
$\|u\|_{\infty}=\max_{j}|u_j|.$ For $d\in\mathbb{N}$, denote $[d]$
to mean $\{1,2,\cdots,d\}.$ Given an index set $T\subset [d]$ and a
matrix $D\in\mR^{n\times d}$, $T^c$ is the complement of $T$ in
$[d]$, $D_T$ is the submatrix of $D$ formed from the columns of $D$
indexed by $T$.\footnote{We note that this notation will
occasionally be abused to refer to the $n\times d$ matrix obtained
by setting the columns of $D$ indexed by $T^c$ to zero. The usage
should be clear from the context, but in most cases there is no
substantive difference between the two.} For a matrix $D_1$, we
write $D_{1T}$ to mean $(D_1)_T$. Write $D^*$ to mean the conjugate
transpose of a matrix $D$, $D^T$ to mean the transpose of $D$, and
$D_T^*$ to mean $(D_T)^*$. For a vector $x\in\mR^d$, $x_{[s]}$
denotes the vector consisting of the $s$ largest entries of $x$ in
magnitude. $C>0$ (or $c$, $c_1$) denotes a universal constant that
might be different in each occurrence. $D^+$ denotes the
Moore-pseudo inverse of a matrix $D,$ and $\ker D$ denotes the null
space of $D$. For a frame $D$ with frame bounds $0<\mathcal{L}\leq
\mathcal{U}<\infty$, $D^{\dag}=(DD^*)^{-1}D$ is the canonical dual
frame. Note that $D^{\dag}=(D^*)^+,$ $D^{\dag}D^*=I,$ and the lower
and upper frame bound of $D^{\dag}$ is given by $1/\mathcal{U}$ and
$1/\mathcal{L}$, respectively. The smallest and largest eigenvalues of a symmetric matrix $B \in \mathbb{R}^{d \times d}$ are denoted by $\lambda_{\min}(B)$ and $\lambda_{\max}(B)$, respectively.

\section{Sparse recovery via $l_q$-analysis}
In this section we prove Theorem \ref{thm2}. We first show that if
the measurement matrix $A$ satisfies an $(D^{\dag},q)$-RIP condition, then
the unknown signal can be approximately recovered by solving the
$l_q$ analysis optimization. The following basic inequalities related to the $l_p$
(quasi)norm are useful for our proofs. For any vectors
$u,v\in\mR^N,$ one has
\be\label{BasicInequality} \|u\|_{p_2} \leq \|u\|_{p_1} \leq N^{1/p_1-1/p_2}\|u\|_{p_2},\quad 0<p_1\leq p_2\leq \infty \ee
and the following triangle inequality for $\|\cdot\|_q^q$ with $q\in(0,1]$:
\be\label{TriangleInequ}\|u+v\|_q^q\leq \|u\|_q^q + \|v\|_q^q .\ee

\subsection{Recovery results based on $(D^{\dag},q)$-RIP}
In this subsection, we give $(D^{\dag},q)$-RIP guarantee results on sparse recovery with frames from noisy measurements $y=Af+z$ via solving
the following $l_q$-analysis optimization \be\label{LQAnalysis}
\argmin_{\tilde{f}\in\mR^n} \|D^*\tilde{f}\|_q\quad \mbox{subject
to}\quad \|A\tilde{f}-y\|_r \leq \varepsilon, \ee
where $0<q\leq 1, 1 \leq r\leq \infty, \varepsilon\geq 0$ and the noise term $z\in\mR^m$ satisfies $\|z\|_r\leq \varepsilon$.
\begin{thm}\label{thm1} Let $0<q\leq 1, 1 \leq r\leq \infty, \varepsilon\geq 0$.
Suppose we observe data from the model $y=Af+z$ with $\|z\|_r\leq \varepsilon.$
  Let $D$ be a frame with frame bounds $0<\mathcal{L}\leq\mathcal{U}<\infty$. Fix positive integers $s,a$ with
  $s<a$.
  Assume that the
  $(D^{\dag},q)$-RIP constant of the measurement matrix $A$ satisfies
   \be\label{e_RIPC}
  \rho^{1-q/2}\left(\rho^{2/q-1}+1\right)^{q/2}\kappa^q(1+\delta_a)<1-\delta_{s+a},\quad
\ee
whree
\be
\label{rhokappa}\rho=\frac{s}{a},\qquad \kappa=\frac{\mathcal{U}}{\mathcal{L}}.
\ee
  Then any solution $\hat{f}$ to (\ref{LQAnalysis}) satisfies
  \be\label{Estimationh} \|\hat{f}-f\|_2\leq
  C_1\frac{\|D^*f-(D^*f)_{[s]}\|_q}{s^{1/q-1/2}} + C_2m^{1/q-1/r} \varepsilon \ee
  and
  \be \label{EstimationDh} \|D^*\hat{f} - D^*f\|_q^q \leq C_3\|D^*f-(D^*f)_{[s]}\|_q^q + C_4a^{1-q/2}m^{1-q/r}\varepsilon^q,\ee
   where $C_1,C_2,C_3$ and $C_4$ are
  positive constants (given explicitly in the proof) depending only on the
  $(D^{\dag},q)$-RIP constant $\delta_s$, $\delta_{s+a}$, $\rho$, $q$, $\mathcal{L}$ and $\kappa$.
\end{thm}
\begin{rem}
\begin{enumerate}
\item For $q=1$, Liu et al. \cite{LML12}
  considered the problem of recovering signals which are
compressible in a general frame $D$ via $l_1$-analysis, with the
assumption that the measurement matrix $A$ satisfies an
$\tilde{D}$-RIP condition, where $\tilde{D}$ is a general dual frame
of $D$. Note that $D^{\dag}$ is the canonical dual frame of $D$, and
Theorem \ref{thm1} might be extended to
$(\tilde{D},q)$-RIP recovery results where $\tilde{D}$ is a general dual frame. However, since in most cases, $A$ is a random matrix, using a general dual frame $\tilde{D}$ would not lead to any advantage.
\item The above theorem requires that the measurement matrix $A$ satisfies an $(D^{\dag},q)$-RIP condition (\ref{e_RIPC}).
From the coming subsection, we will see that such a condition is met by setting $a=O(s)$ and $m=O(s+qs\log n).$
In this case, the last term of (\ref{Estimationh}) is roughly $C(s+qs\log n)^{1/q-1/r}.$
\end{enumerate}
\end{rem}
The proof of the above theorem involves several lemmas. We postpone the proof in Subsection 2.4. A direct consequence
of the theorem is the following corollary, which is useful for the proof of Theorem \ref{thm2}
\begin{corollary}\label{Cor1}
  Under the assumptions of Theorem \ref{thm1}, we further assume that $z=0$.
  Then any solution $\hat{f}$ to ($P_q$) satisfies
  $$ \|\hat{f}-f\|_2\leq
  C_1\frac{\|D^*f-(D^*f)_{[s]}\|_q}{s^{1/q-1/2}} ,$$
  where $C_1$ is given by Theorem \ref{thm1}.
\end{corollary}
\subsection{Random Gaussian measurements implying $(D^{\dag},q)$-RIP }
We next determine how many random Gaussian observations are needed
to guarantee that the $(D^{\dag},q)$-RIP condition in Theorem
\ref{thm1} holds with high probability. Let $A$ be an $m\times n$
matrix whose entries are i.i.d Gaussian random variables
$\mathcal{N}(0,\sigma^2).$ For a given $q$, let
$\varrho_q=\sigma^{q}2^{q/2}\Gamma(\frac{q+1}{2})/\sqrt{\pi}.$ Using a
same argument as that for \cite[Lemma 3.2]{CS08}, one can prove the
following result.
\begin{lemma}\label{lemma1}
Let $0<q\leq 1$, $A$ be an $m\times n$ matrix whose entries are
i.i.d Gaussian random variables $\mathcal{N}(0,\sigma^2).$ Then for
any fixed $x\in\mR^n,$ $\eta>0,$
\be\label{e6}\mathbb{P}\left(\left|\|Ax\|_q^q-
m\varrho_q\|x\|_2^q\right|\geq\eta
m\varrho_q\|x\|_2^q\right)\leq2\exp{\left(-\frac{\eta^2m}{2q\beta_q^2}\right)},\ee
where
$$\varrho_q=\sigma^{q}2^{q/2}\Gamma(\frac{q+1}{2})/\sqrt{\pi},$$ and
\be\label{e7}\beta_q=(31/40)^{1/4}\left[1.13+\sqrt{q}
\left(\frac{\Gamma\left(\frac{q+1}{2}\right)}{\sqrt{\pi}}\right)^{-1/q}\right].\ee\end{lemma}
\begin{rem}
Note that
$\left(\Gamma\left(\frac{q+1}{2}\right)/\sqrt{\pi}\right)^{1/q}$ is
an increasing function of $q$, bounded below by
$\mathrm{e}^{-\gamma/2}/2\simeq 0.375$ \cite{CS08}. Therefore
$$\beta_1=(31/40)^{1/4}(1.13+\sqrt{\pi})\simeq 3.8407,$$ and when $q\rightarrow
0,$ $\beta_q\rightarrow 1.13(31/40)^{1/4}\simeq 1.0602.$
\end{rem} Using a same argument as that for \cite[Lemma 3.3]{CS08}, one can
prove the following result Lemma 2.4. In this paper, we provide
alterative simple proof for this result. Such a proof is motivated
by \cite{BDDW08}.
\begin{lemma}\label{lemma2}
Let $0<q\leq 1, \eta,\epsilon>0$, $A$ be an $m\times n$ matrix whose
entries are i.i.d Gaussian random variables
$\mathcal{N}(0,\sigma^2).$ Set
$\delta=\frac{\eta+\epsilon^q}{1-\epsilon^q}$. Then
\be\label{e8}(1-\delta)m\varrho_q\|Dv\|_2^q\leq\|ADv\|_q^q\leq(1+\delta)m\varrho_q\|Dv\|_2^q\ee
holds uniformly for all $k$ sparse vectors $v\in\mR^d$ with
probability exceeding
\be\label{e9}1-2\left(\frac{3\mathrm{e}d}{\epsilon
k}\right)^k\exp{\left(-\frac{\eta^2m}{2q\beta_q^2}\right)}.\ee
\end{lemma}
\begin{proof}
First note that it suffices to prove (\ref{e8}) in the case of
$\|Dv\|_2=1,$ since $A$ is linear. Let $\Sigma_k=\{Dv:\|v\|_0\leq
k\}.$ Fix an index set $T\subset[d]$ with $|T|=k,$ denote by $X_T$
the subspace spanned by the columns of $D_T.$ Note that $X_T$ is at
most $k$ dimensional, and we endow the $l_2$ norms. Choose a finite
$\epsilon$ covering of the unit sphere in $X_T$, i.e., a set of
points $Q_T\subset X_T$, with $\|u\|_2=1$ for all $u\in Q_T$, such
that for all $v\in X_T,\|v\|_2=1$, we have
$$\min_{u\in Q_T}\|v-u\|_2\leq \epsilon.$$ According to \cite[Lemma
2.2]{MPT08}, there exists such an $Q_T$ with $|Q_T|\leq
\left(\frac{3}{ \epsilon}\right)^k$. Repeat this process for each
possible index set $T$, and collect all the sets $Q_T$ together:
$$Q=\bigcup_{T: |T|=k} Q_T.$$ Since the number of possible $T$ is
$\binom{d}{k}$, thus, by Sterling's approximation, $$|Q|\leq
\binom{d}{k}\left(\frac{3}{ \epsilon}\right)^k\leq
\left(\frac{\mathrm{e}d}{k}\right)^k \left(\frac{3}{
\epsilon}\right)^k=\left(\frac{3\mathrm{e}d}{\epsilon k}\right)^k.$$ Applying
Lemma \ref{lemma1}, one gets that
$$\mathbb{P}\left(\max_{u\in Q}\left|\|Au\|_q^q- m\varrho_q\|u\|_2^q\right|\geq\eta
m\varrho_q\|u\|_2^q\right)\leq2\left(\frac{3\mathrm{e}d}{\epsilon
k}\right)^k\exp{\left(-\frac{\eta^2m}{2q\beta_q^2}\right)}.$$ It
thus follows that with probability exceeding (\ref{e9}),
\be\label{e10} m\varrho_q(1-\eta)\|u\|_2^q\leq\|Au\|_q^q\leq
m\varrho_q(1+\eta)\|u\|_2^q\quad \mbox{for all }u\in Q.\ee

Now define $B$ as the smallest number such that $$ \|Av\|_q^q\leq
m\varrho_q(1+B)\|v\|_2^q\quad \mbox{for all }v\in\Sigma_k, \|v\|_2=1.$$
Our goal is to show that $B\leq \delta.$ Note that from the
definitions of $Q_T$ and $Q$, we know that for any $v\in\Sigma_k,
\|v\|_2=1$, we can choose an $u\in Q$ such that $\|u-v\|_2\leq
\epsilon$ and such that $u-v\in\Sigma_k.$ Thus, we get \bea
\|Av\|_q^q\leq \|A(v-u)\|_q^q+\|Au\|_q^q\leq
m\varrho_q[(1+B)\epsilon^q+1+\eta].\eea It thus follows from the
definition of $B$ that
$$m\varrho_q(1+B)\leq m\varrho_q[(1+B)\epsilon^q+1+\eta],$$ which leads to
$B\leq \delta.$ The lower inequality follows from this since by
triangle inequality for $\|\cdot\|_q^q,$ \bea \|Av\|_q^q \geq
\|Au\|_q^q-\|A(u-v)\|_q^q\geq m\varrho_q[(1-\eta)-(1+\delta)\epsilon^q]=
m\varrho_q(1-\delta).\eea The proof is finished.
\end{proof}
\subsection{Proof of Theorem \ref{thm2}}
Now we are ready to prove Theorem \ref{thm2}. The proof is similar to
\cite{CS08}, with a  simple modification.
  By Corollary \ref{Cor1}, we only need to prove that the ($D^{\dag},q$)-RIP condition (\ref{e_RIPC})
  holds with high probability.
  Let $a=ts=\lceil(2^{q/2}b\kappa^q)^{\frac{2}{2-q}}\rceil s$ and
  $k=a+s=(\lceil(2^{q/2}b\kappa^q)^{\frac{2}{2-q}}\rceil+1)s,$ where
  $b>1.$ Note that $$\rho^{1-q/2}\left(\rho^{2/q-1}+1\right)^{q/2}\kappa^q
  <\rho^{1-q/2}2^{q/2}\kappa^q\leq b^{-1}.$$
  Therefore, a sufficient condition for (\ref{e_RIPC}) is
  $\delta_{(t+1)s}<\frac{b-1}{b+1}.$  Choose $\eta=r(b-1)/(b+1)$ for
  $r\in(0,1)$ and $\epsilon^q=(1-r)(b-1)/(2b)$. We have $(\eta+\epsilon^q)/(1-\epsilon^q)\leq(b-1)/(b+1).$
  By Lemma \ref{lemma2}, $A$ will fail to satisfy
  ($D^{\dag},q$)-RIP (\ref{e8}) with probability less than $$2\left(\frac{3\mathrm{e}d}{\epsilon
k}\right)^k\exp{\left(-\frac{\eta^2m}{2q\beta_q^2}\right)}.$$ It is
enough to prove that the above quantity can be bounded by
$(\frac{s}{ed})^s\leq 1/\binom{d}{s}.$ This is equivalent to
 \bea m&\geq&\frac{2q\beta_q^2}{\eta^2}\left[k\ln
\left(\frac{3\mathrm{e}d}{k}\right)+k\ln\left(\frac{1}{\epsilon}\right)+\ln2+s\ln \left(\frac{\mathrm{e}d}{s}\right)\right]\\
&=&\frac{2q\beta_q^2(b+1)^2}{r^2(b-1)^2}\left[\left(\lceil(2^{q/2}b\kappa^q)^{\frac{2}{2-q}}\rceil+1\right)s\left(\ln
\left(\frac{\mathrm{e}d}{s}\right)+\ln
3-\ln\left(\lceil(2^{q/2}b\kappa^q)^{\frac{2}{2-q}}\rceil+1\right)\right)\right.\\&&
+\left.
\frac{1}{q}(\lceil(2^{q/2}b\kappa^q)^{\frac{2}{2-q}}\rceil+1)s\ln\left(\frac{2b}{(1-r)(b-1)}\right)
+\ln2+s\ln \left(\frac{\mathrm{e}d}{s}\right)\right].\eea Similar to
\cite{CS08}, by setting $r=0.849$ and $b=5$, \bea m&\geq&
6.25q\beta_q^2\left[
\left(\lceil(5\cdot2^{q/2}\kappa^q)^{\frac{2}{2-q}}\rceil+1\right)\left(\ln
3-\ln\left(\lceil(5\cdot2^{q/2}\kappa^q)^{\frac{2}{2-q}}\rceil+1\right)\right)s+\ln
2\right.\\
&&\left.+\left(\lceil(5\cdot2^{q/2}\kappa^q)^{\frac{2}{2-q}}\rceil+2\right)s\ln
\left(\frac{\mathrm{e}d}{s}\right)\right]
+17.6\beta_q^2(\lceil(5\cdot2^{q/2}\kappa^q)^{\frac{2}{2-q}}\rceil+1)s
.\eea The proof is finished.
\subsection{Proof of Theorem \ref{thm1}}
In this subsection, we prove Theorem \ref{thm1}. Our goal is to
bound the norm of $h=\hat{f}-f$,  where $\hat{f}$ and $f$ are as in the theorem.

We begin by establishing several lemmas for a general vector $h.$ For arbitrary fixed $h\in \mR^n,$ since $D$ is a frame for $\mR^n$,
one can upper bound $\|h\|_2^2$ by $\mathcal{L}^{-1}\|D^*h\|_2^2.$
To estimate $\|D^*h\|_2,$ we use a common decomposition technique in the standard compressed sensing (e.g., \cite{CRT06b}).
We write $D^*h=(x_1,\cdots,x_d)^T.$
Rearranging the indices if necessary, we may assume that
$|x_{1}|\geq|x_{s+2}|\geq\cdots\geq|x_d|.$ Set
$T=T_{0}=\{1,2,\cdots,s\},$ $T_{1}=\{s+1,s+2,\cdots,s+a\},$ and
$T_{i}=\{s+(i-1)a+1,\cdots,s+ia\},i=2,\cdots,$ with the last subset
of size less than or equal to $a$. Denote $T_{01}=T_{0}\cup T_{1}.$
 Note that
by applying the first inequality of (\ref{BasicInequality}), we have
\be\label{FinalEstimate1} \mathcal{L}^{q/2}\|h\|_2^q \leq\|D^*h\|_2^q = \left(\|D_{T_{01}}^*h\|_2^2+\sum_{j\geq2}\|D_{T_{j}}^*h\|_2^2\right)^{q/2} \leq\|D_{T_{01}}^*h\|_2^q+\sum_{j\geq2}\|D_{T_{j}}^*h\|_2^q. \ee
In what follows, we shall upper bound the last two terms.
The following lemma, which was originally proved in \cite{CS08}, gives an upper bound of the tail $\sum_{j\geq2}\|D_{T_{j}}^*h\|_2^q$ in terms on $\|D_{T^c}^*h\|_q^q$.
We prove it for completeness.
\begin{lemma}[Bounding the tail]\label{TailEstimation} We have
\be\label{e_sum}\sum_{j\geq2}\|D_{T_j}^*h\|_2^q\leq
a^{q/2-1}\|D_{T^c}^*h\|_q^q. \ee
\end{lemma}
\begin{proof} Fix $i>0$, for each $l\in T_i$ and $l'\in T_{i+1}$, obviously we
have $|x_{l'}|^q\leq |x_{l}|^q.$ Thus, $|x_{l'}|^q\leq
\|D_{T_i}^*h\|_q^q/a.$ It thus follows that
 \be\sum_{j\geq2}\|D_{T_j}^*h\|_2^q\leq
a^{q/2-1}\sum_{j\geq1}\|D_{T_j}^*h\|_q^q=
a^{q/2-1}\|D_{T^c}^*h\|_q^q.\nonumber \ee
\end{proof}
To bound $\|D_{T_{01}}^*h\|_2^q$, we need the following result which utilizes the fact that $A$ satisfies the $(D^{\dag},q)$-RIP.
\begin{lemma}[Consequence of the $(D^{\dag},q)$-RIP]\label{LemDQRIP} Assume that $A$
 satisfies the $(D^{\dag},q)$-RIP of order $s+a.$
Let
\be\label{delta}\Delta=\frac{1+\delta_{a}}{1-\delta_{s+a}}.\ee Then, we have
  \be\label{e1}\|DD_{T_{01}}^*h\|_2^2\leq \kappa\mathcal{U}\Delta^{2/q}a^{1-2/q} \left(\|D_{T^c}^*h\|_q^q + { \mathcal{L}^{q/2}a^{1-q/2}\|Ah\|_q^q \over 1+ \delta_{a}} \right)^{2/q}.\ee
\end{lemma}
\begin{proof}
Note that $D^{\dag}D^*=I$, thus,
\bea\|Ah\|_q^q&=&\|AD^{\dag}D^*h\|_q^q=\left\|AD^{\dag}D_{T_{01}}^*h+\sum_{j\geq2}AD^{\dag}D_{T_j}^*h\right\|_q^q\\
&\geq&\|AD^{\dag}D_{T_{01}}^*h\|_q^q-\sum_{j\geq2}\|AD^{\dag}D_{T_j}^*h\|_q^q,\eea
where the last inequality follows from the triangle inequality (\ref{TriangleInequ}).
 It thus follows from the definitions of $(D^{\dag},q)$-RIP that
  \bea \|Ah\|_q^q \geq
(1-\delta_{s+a})\|D^{\dag}D_{T_{01}}^*h\|_2^q-
(1+\delta_{a})\sum_{j\geq2}\|D^{\dag}D_{T_j}^*h\|_2^q. \eea
Using the definition of frame (\ref{Frames}), which is equivalent to
\be\label{frameEquivalent}
\mathcal{L} \leq \lambda_{\min}(DD^*) \leq \lambda_{\max}(DD^*) \leq \mathcal{U},
\ee
and implies
\bea
\mathcal{U}^{-1} \leq \lambda_{\min}((DD^*)^{-1}) \leq \lambda_{\max}((DD^*)^{-1}) \leq \mathcal{L}^{-1},
\eea
and recalling that $D^{\dag} = (DD^*)^{-1} D,$ we get
\bea
\|Ah\|_q^q
&\geq&\mathcal{U}^{-q}(1-\delta_{s+a})\|DD_{T_{01}}^*h\|_2^q-\mathcal{L}^{-\frac{q}{2}}(1+\delta_{a})\sum_{j\geq2}\|D_{T_j}^*h\|_2^q.\eea
Introducing (\ref{e_sum}) to the above, \bea \|Ah\|_q^q \geq
\mathcal{U}^{-q}(1-\delta_{s+a})\|DD_{T_{01}}^*h\|_2^q-\mathcal{L}^{-\frac{q}{2}}(1+\delta_{a})a^{q/2-1}\|D_{T^c}^*h\|_q^q.\eea
Rearranging terms, noting $\delta$ and $\kappa$ are given by (\ref{delta}) and (\ref{rhokappa}) respectively
, we get that
\bea \|DD_{T_{01}}^*h\|_2^q &\leq&
\Delta\kappa^{q/2}\mathcal{U}^{q/2}a^{q/2-1}\|D_{T^c}^*h\|_q^q + \mathcal{U}^q\|Ah\|_q^q/(1-\delta_{s+a})\\
&=&\Delta\kappa^{q/2}\mathcal{U}^{q/2}a^{q/2-1} \left(\|D_{T^c}^*h\|_q^q + { \mathcal{L}^{q/2}a^{1-q/2}\|Ah\|_q^q \over 1+ \delta_{a}} \right). \eea
Taking the $(2/q)$-th power on both sides, we get the desired result.\end{proof}
With the estimation on $\|DD_{T_{01}}^*h\|_2^{2},$ we are ready to give an upper bound on
$\|D_{T_{01}}^*h\|_2^{q}$.
This can be done by developing a relationship between $\|DD_{T_{01}}^*h\|_2^{2}$ and $\|D_{T_{01}}^*h\|_2^{2}.$
\begin{lemma}[Bounding $\|D_{T_{01}}^*h\|_2^{q}$]\label{LemDT01} Under the assumptions and notations of Lemma \ref{LemDQRIP} , we have
  \be \label{DT01}\|D_{T_{01}}^*h\|_2^{q} \leq 2^{-q/2}\left(1 + \sqrt{ 1 + 4 \kappa^{-2}\Delta^{-2 / q}} \right)^{q/2}\kappa^q\Delta a^{q/2-1} \left(\|D_{T^c}^*h\|_q^q + { \mathcal{L}^{q/2}a^{1-q/2}\|Ah\|_q^q \over 1+ \delta_{a}} \right).\ee
\end{lemma}
\begin{proof}
Note that by Cauchy-Schwarz inequality and (\ref{Frames}) \bea
\|D_{T_{01}}^*h\|_2^{4}=|\la h,DD_{T_{01}}^*h\ra
|^2\leq\|h\|_2^2\|DD_{T_{01}}^*h\|_2^{2} \leq \mathcal{L}^{-1}\|D^*h\|_2^2\|DD_{T_{01}}^*h\|_2^{2}. \eea
Substituting with
$\|D^*h\|_2^2=\|D_{T_{01}}^*h\|_2^{2}+\sum_{j\geq2}\|D_{T_j}^*h\|_2^2,$
and then applying the first inequality of (\ref{BasicInequality}) to upper bound the term
$\sum_{j\geq2}\|D_{T_j}^*h\|_2^2,$ we get \be\label{relationship}
\|D_{T_{01}}^*h\|_2^{4}\leq
\mathcal{L}^{-1}\|DD_{T_{01}}^*h\|_2^{2}
\left(\|D_{T_{01}}^*h\|_2^{2}+\left(\sum_{j\geq2}\|D_{T_j}^*h\|_2^q\right)^{2/q}\right) .\ee
Introducing (\ref{e_sum}) to the above, \bea
\|D_{T_{01}}^*h\|_2^{4}\leq
\mathcal{L}^{-1}\|DD_{T_{01}}^*h\|_2^{2}
\left(\|D_{T_{01}}^*h\|_2^{2}+a^{1-2/q}\|D_{T^c}^*h\|_q^2\right).\eea
Rearranging terms and completing the square, this reads as
\bea
\left(\|D_{T_{01}}^*h\|_2^{2} - \frac{\|DD_{T_{01}}^*h\|_2^{2}}{2\mathcal{L}}\right)^2
\leq \|DD_{T_{01}}^*h\|_2^{2}
\left(\frac{\|DD_{T_{01}}^*h\|_2^{2}}{4\mathcal{L}^{2}}+\frac{\|D_{T^c}^*h\|_q^2}{\mathcal{L}a^{2/q-1}}\right).\eea
Taking square root of each side and rearranging terms, we get
$$\|D_{T_{01}}^*h\|_2^{2} \leq \frac{\|DD_{T_{01}}^*h\|_2^{2}}{2\mathcal{L}} +
  \sqrt{\|DD_{T_{01}}^*h\|_2^{2}
\left(\frac{\|DD_{T_{01}}^*h\|_2^{2}}{4\mathcal{L}^{2}}+\frac{\|D_{T^c}^*h\|_q^2}{\mathcal{L}a^{2/q-1}}\right)}.$$
Recalling  $\delta$ and $\kappa$ are given by (\ref{delta}) and (\ref{rhokappa}) respectively, and upper bounding the term $\frac{\|D_{T^c}^*h\|_q^2}{\mathcal{L}a^{2/q-1}}$
by $$ {a^{1-2/q}(\|D_{T^c}^*h\|_q^q)^{2/q} \over \mathcal{L}} \leq \frac{\kappa\mathcal{U}\Delta^{2/q}a^{1-2/q} \left(\|D_{T^c}^*h\|_q^q + { \mathcal{L}^{q/2}a^{1-q/2}\|Ah\|_q^q \over 1+ \delta_{a}} \right)^{2/q}} {\kappa\mathcal{U}\mathcal{L} \Delta^{2/q}},$$
and applying (\ref{e1}), we have
\bea \|D_{T_{01}}^*h\|_2^{2} &\leq& \left({1 \over 2\mathcal{L}} + \sqrt{ {1\over 4\mathcal{L}^2} + {1 \over \kappa\mathcal{UL}\Delta^{2 / q}}} \right)\kappa\mathcal{U}\Delta^{2/q}a^{1-2/q} \left(\|D_{T^c}^*h\|_q^q + { \mathcal{L}^{q/2}a^{1-q/2}\|Ah\|_q^q \over 1+ \delta_{a}} \right)^{2/q}\\
&=&2^{-1}\left(1 + \sqrt{ 1 + 4 \kappa^{-2}\Delta^{-2 / q}} \right)\kappa^2\Delta^{2/q}a^{1-2/q} \left(\|D_{T^c}^*h\|_q^q + { \mathcal{L}^{q/2}a^{1-q/2}\|Ah\|_q^q \over 1+ \delta_{a}} \right)^{2/q}.\eea
Taking the $(q/2)$-th power on both sides, we get our desired result.
\end{proof}
By lemmas \ref{TailEstimation} and \ref{LemDT01}, we know that the last two terms of (\ref{FinalEstimate1}) can be upper bounded in terms of $\|D_{T^c}^*h\|_q^q$ and $\|Ah\|_q^q.$
In what follows, we develop another two lemmas to bound $\|D_{T^c}^*h\|_q^q$ and $\|Ah\|_q^q.$ The following lemma shows that an suitable $(D^{\dag},q)$-RIP condition implies the robust $D$-NSP$_q$ of the measurement matrix $A$.
\begin{lemma}[Robust $D$-NSP$_q$]\label{LemNSP}
Under the assumptions and notations of Lemma \ref{LemDQRIP}, we have for any $h\in \mR^N,$
\be\label{e_DNSP}
\|D_T^*h\|_q^{q} \leq \theta \|D_{T^c}^*h\|_q^q +  { \theta\mathcal{L}^{q/2}a^{1-q/2}\|Ah\|_q^q \over 1+ \delta_{a}}  ,\ee where
\be\label{e_thet}\theta=2^{-q/2}\left(1 + \sqrt{ 1 + 4 \kappa^{-2} \Delta^{-2 / q}} \right)^{q/2}\kappa^q\Delta\rho^{1-q/2},\ee
 where $\rho$ is given by (\ref{rhokappa}).
 In particular, if the condition (\ref{e_RIPC}) is satisfied, then $\theta<1.$
\end{lemma}
\begin{proof}
Using
$\|D_{T_{01}}^*h\|_2\geq\|D_{T_{0}}^*h\|_2$ and then applying H\"{o}lder's inequality mentioned in (\ref{BasicInequality})
to lower bound $\|D_{T_{0}}^*h\|_2$, we get $\|D_{T_{01}}^*h\|_2^q \geq s^{q/2 - 1}\|D_T^*h\|_q^{q}.$ Therefore, combining with (\ref{DT01}), we get
\bea  s^{q/2 - 1}\|D_T^*h\|_q^{q} \leq 2^{-q/2}\left(1 + \sqrt{ 1 + 4 \kappa^{-2}\Delta^{-2 / q}} \right)^{q/2}\kappa^q\Delta a^{q/2-1} \left(\|D_{T^c}^*h\|_q^q + { \mathcal{L}^{q/2}a^{1-q/2}\|Ah\|_q^q \over 1+ \delta_{a}} \right).\eea
Dividing both sides by
$s^{1-2/q}$, we get (\ref{e_DNSP}). It remains to show that $\theta<1.$ Actually, this can be verified by
showing that $\theta^{2/q}<1$, that is
$$ 2^{-1}\left(1 + \sqrt{ 1 + 4\kappa^{-2}\Delta^{-2 / q}} \right)\kappa^2\Delta^{2/q}\rho^{2/q-1} < 1. $$
Multiplying  both sides by $\sqrt{ 1 + 4\kappa^{-2}\Delta^{-2 / q}} - 1$,
this reads as
$$ 2^{-1}\cdot 4\kappa^{-2}\Delta^{-2 / q}\cdot\kappa^2\Delta^{2/q}\rho^{2/q-1} < \sqrt{ 1 + 4\kappa^{-2}\Delta^{-2 / q}} - 1, $$
which is
equivalent to $$ 2\rho^{2/q-1} + 1 <
\sqrt{1+4\Delta^{-2/q}\kappa^{-2}}. $$ Taking the second power of
both sides, subtracting both sides by $1$ and by a simple
calculation, this reads as
$$ 4\rho^{2/q-1}(\rho^{2/q-1} + 1)< 4 \Delta^{-2/q}\kappa^{-2}.$$
Dividing both sides by $4 \Delta^{-2/q}\kappa^{-2}$, and then taking
the $(q/2)$-th power on both sides, we know that this is equivalent
to $$ \rho^{1-q/2}(\rho^{2/q-1} + 1)^{q/2}\Delta\kappa^{q}< 1.$$
Introducing (\ref{delta}), we know that the
above inequality is equivalent to (\ref{e_RIPC}). Consequently,
(\ref{e_RIPC}) implies $\theta<1.$
\end{proof}
Note that the above lemmas hold for any $h\in \mR^n.$ In what
follows, we shall choose $h=\hat{f}-f$, where $\hat{f}$ is a
solution of $(\ref{LQAnalysis})$ and $f$ is the original signal. Let
$\Omega$ be the index set of the largest $s$ entries of $ D^*f$ in
magnitude. The following results can be verified by using the fact
that $\hat{f}$ is a solution of $(\ref{LQAnalysis})$.

\begin{lemma}[Consequence of a solution]\label{LemSolution} Let $h=\hat{f}-f$, where $\hat{f}$ is a solution of $(\ref{LQAnalysis})$ and
$f$ satisfies $\|Af-y\|_r \leq \varepsilon$. Let $\Omega$ be the
index set of the largest $s$ entries of $ D^*f$ in magnitude. Then
we have \be\label{errorbound} \|Ah\|_q^q\leq  m^{1-q/r}
(2\varepsilon)^q \ee and \be\label{e_cone constraint1}
\|D^*_{T^c}h\|_q^q\leq\|D^*_{T}h\|_q^q+2\|D^*_{\Omega^c}f\|_q^q.\ee
\end{lemma}
\begin{proof}
Since both $\hat{f}$ and $f$ are feasible and $r\geq 1$, we have
$$\|Ah\|_r\leq \|Af-y\|_r + \|A\hat{f}-y\|_r \leq  2\varepsilon.$$
Using the H\"{o}lder's inequality mentioned in (\ref{BasicInequality}), we get
$$\|Ah\|_q^q \leq m^{1-q/r} \|Ah\|_r^q \leq m^{1-q/r} (2\varepsilon)^q.$$ This proves (\ref{errorbound}).

Since $\hat{f}$ is a minimizer of $(\ref{LQAnalysis})$, one gets that
$$\|D^*f\|_q^q\geq\|D^*\hat{f}\|_q^q.$$ That is
$$\|D^*_{\Omega}f\|_q^q+\|D^*_{\Omega^c}f\|_q^q\geq\|D^*_{\Omega}\hat{f}\|_q^q+\|D^*_{\Omega^c}\hat{f}\|_q^q.$$
Substituting with $\hat{f}=f+h$ and using the triangle
inequality (\ref{TriangleInequ}),
\bea\|D^*_{\Omega}f\|_q^q+\|D^*_{\Omega^c}f\|_q^q\geq\|D^*_{\Omega}f\|_q^q-\|D^*_{\Omega}h\|_q^q
+\|D^*_{\Omega^c}h\|_q^q-\|D^*_{\Omega^c}f\|_q^q.\eea Rearranging
terms, \bea\|D^*_{\Omega^c}h\|_q^q - \|D^*_{\Omega}h\|_q^q \leq
2\|D^*_{\Omega^c}f\|_q^q. \eea Combining with the fact that \bea
\|D^*_{T^c}h\|_q^q-\|D^*_{T}h\|_q^q\leq
\|D^*_{\Omega^c}h\|_q^q-\|D^*_{\Omega}h\|_q^q, \eea and rearranging
terms, we get (\ref{e_cone constraint1}). The proof is finished.
\end{proof}

We may now conclude the proof of Theorem \ref{thm1}. We first apply lemmas \ref{LemNSP} and \ref{LemSolution}
to get an upper bound on $\|D^*_{T^c}h\|_q^q$.
Introducing (\ref{e_DNSP}) to (\ref{e_cone constraint1}), we get
$$
\|D^*_{T^c}h\|_q^q\leq\|D^*_{T}h\|_q^q+2\|D^*_{\Omega^c}f\|_q^q \leq \theta \|D_{T^c}^*h\|_q^q +  { \theta\mathcal{L}^{q/2}a^{1-q/2}\|Ah\|_q^q \over 1+ \delta_{a}} + 2\|D^*_{\Omega^c}f\|_q^q. $$
Rearranging terms and dividing both sides by $1-\theta$ (noting that $\theta<1$ by Lemma \ref{LemNSP}), we get
\be\label{e5}\|D^*_{T^c}h\|_q^q\leq { \theta\mathcal{L}^{q/2}a^{1-q/2}\|Ah\|_q^q \over (1-\theta)(1+ \delta_{a})} + \frac{2}{1-\theta}\|D^*_{\Omega^c}f\|_q^q.\ee
Now we can upper bound $\|h\|_2.$
Introducing (\ref{e_sum}) and (\ref{DT01}) to (\ref{FinalEstimate1}), and noting that $\theta$ and $\rho$ are given by (\ref{e_thet}) and (\ref{rhokappa}) respectively, we get
\bea \mathcal{L}^{q/2}\|h\|_2^q &\leq& \theta
s^{q/2-1} \left(\|D_{T^c}^*h\|_q^q + {
\mathcal{L}^{q/2}a^{1-q/2}\|Ah\|_q^q \over 1+ \delta_{a}} \right)
+ a^{q/2-1}\|D_{T^c}^*h\|_q^q \\
&=& \frac{(\theta+\rho^{1-q/2}) \|D_{T^c}^*h\|_q^q}{s^{1-q/2}} +
{\theta \mathcal{L}^{q/2} \rho^{q/2-1} \|Ah\|_q^q \over 1+ \delta_{a}}. \eea
Applying (\ref{e5}) to the above, we get
\bea \mathcal{L}^{q/2} \|h\|_2^q &\leq&
\frac{(\theta+\rho^{1-q/2})}{s^{1-q/2}}\left( { \theta\mathcal{L}^{q/2}a^{1-q/2}\|Ah\|_q^q \over (1-\theta)(1+ \delta_{a})} + \frac{2}{1-\theta}\|D^*_{\Omega^c}f\|_q^q\right)  +
{\theta \mathcal{L}^{q/2} \rho^{q/2-1} \|Ah\|_q^q \over 1+ \delta_{a}}\\
&=&\frac{2(\theta+\rho^{1-q/2})\|D^*_{\Omega^c}f\|_q^q}{(1-\theta)s^{1-q/2}}
+ {\theta \mathcal{L}^{q/2} (1+\rho^{q/2-1})\|Ah\|_q^q \over
(1-\theta)(1+ \delta_{a})} . \eea Using (\ref{errorbound}) to the
above, and dividing both sides by $\mathcal{L}^{q/2}$, \bea
\|h\|_2^q \leq
\frac{2(\theta+\rho^{1-q/2})\|D^*_{\Omega^c}f\|_q^q}{\mathcal{L}^{q/2}(1-\theta)s^{1-q/2}}
+ {\theta (1+\rho^{q/2-1}) m^{1-q/r} (2\varepsilon)^q \over
(1-\theta)(1+ \delta_{a})} .\eea Taking the $(1/q)$-th power on both
sides and then using a basic inequality $(b+c)^{1/q}\leq
2^{1/q-1}(b^{1/q}+c^{1/q}),\forall b,c\geq 0$ we get \bea  \|h\|_2
\leq
\frac{(2\theta+2\rho^{1-q/2})^{1/q}\|D^*_{\Omega^c}f\|_q}{\sqrt{\mathcal{L}}(1-\theta)^{1/q}s^{1/q-1/2}}
+ { (2\theta+2\theta\rho^{q/2-1})^{1/q} m^{1/q-1/r} \varepsilon
\over (1-\theta)^{1/q}(1+ \delta_{a})^{1/q}} .\eea Thus, we get
(\ref{Estimationh}). It remains to prove (\ref{EstimationDh}). By
(\ref{e_DNSP}), $$\|D^*h\|_q^q = \|D_T^*h\|_q^q
+¡¡\|D_{T^c}^*h\|_q^q \leq (1+\theta)\|D_{T^c}^*h\|_q^q +  {
\theta\mathcal{L}^{q/2}a^{1-q/2}\|Ah\|_q^q \over 1+ \delta_{a}}. $$
Introducing (\ref{e5}) and then using (\ref{errorbound}), we get
\bea \|D^*h\|_q^q &\leq& (1+\theta)\left({
\theta\mathcal{L}^{q/2}a^{1-q/2}\|Ah\|_q^q \over (1-\theta)(1+
\delta_{a})} + \frac{2}{1-\theta}\|D^*_{\Omega^c}f\|_q^q\right)  + {
\theta\mathcal{L}^{q/2}a^{1-q/2}\|Ah\|_q^q \over 1+ \delta_{a}}
\\ &\leq & \frac{2(1+\theta)}{1-\theta}\|D^*_{\Omega^c}f\|_q^q + {2^{q+1}\theta\mathcal{L}^{q/2}a^{1-q/2}m^{1-q/r}\varepsilon^q\over (1-\theta)(1+ \delta_{a})}.\eea
Thus, we get the desired result (\ref{EstimationDh}).The proof is finished.
\begin{rem}
   In the proof, we derive an upper bound for $\|D^*\hat{f}-D^*f\|_2:$ $$\|D^*\hat{f}-D^*f\|_2\leq
  C_5\frac{\|D^*f-(D^*f)_{[s]}\|_q}{s^{1/q-1/2}} +¡¡C_6m^{1/q-1/r} \varepsilon .$$
\end{rem}
\section{Compressed data separation via $l_q$ split analysis}
In this section, we prove Theorem \ref{thm3}.
The proof is similar to that of Theorem \ref{thm2}, and makes use
some ideas from \cite{LLS13TIT}. We first establish an $(D,q)$-RIP recovery
result for compressed data separation with $\iota$ components ($\iota \geq 2$). Considering $\iota =2$, we then utilize Lemma \ref{lemma2} to show that
such an $(D,q)$-RIP condition holds with high probability.
As a result, one can finish the proof.

Let $\iota$ be a positive integer greater than $2$ and  $D_1\in \mR^{n \times d_1}, D_2\in \mR^{n \times d_2},\cdots, D_{\iota} \in \mR^{n \times d_{\iota}}$
be $\iota$ tight frames with frames bounds $1$ for  $\mR^{n}$. Set $s=s_1+s_2+ \cdots + s_{\iota}$ and $\bar{d}=d_1+d_2 + \cdots + d_{\iota}.$
Let
\be\label{e-notation}\bar{D}=[D_1|D_2|\cdots|D_{\iota}], ~\Psi=
 \begin{pmatrix}
 D_1&     &       & & \\
    & D_2 &       & & \\
    &     & \ddots&  & \\
    &     &       & D_{\iota}
 \end{pmatrix},
\mbox{ and } f=\left(\begin{array}{c}f_1\\f_2 \\ \vdots \\ f_{\iota}\end{array}\right).\ee
Note that from the definition of tight frames,
$$\sum_{j=1}^{\iota} f_j =\sum_{j=1}^{\iota} D_jD_j^*f_j =  \bar{D} \Psi^*f.$$
Then, with $\iota=2$, (\ref{SAA})
can be rewritten  as \be\label{SAAI}\hat{f}=
\argmin\limits_{\tilde{f}\in{\mR^{\iota n}}}\|\Psi^*\tilde{f}\|_q^q\quad\mbox{s.t.}\quad
A\bar{D}\Psi^*\tilde{f}=y.\ee
We define the mutual coherence between $D_1,D_2\cdots,D_{\iota}$ as follows.
\begin{definition} Let $D_1=(d_{1i})_{1\leq i\leq d_1}, D_2=(d_{2j})_{1\leq j\leq d_2},\cdots,D_{\iota} = (d_{\iota j})_{ 1\leq j\leq d_{\iota}}$.
The mutual coherence between $D_1,D_2\cdots,D_{\iota}$ is defined as
$$\mu_1=\mu_1(D_1;D_2;\cdots;D_{\iota})=\max_{k\neq l}\max_{i,j}|\la d_{ki},d_{lj}\ra|.$$\end{definition}
\subsection{$(\bar{D},q)$-RIP recovery result for $l_q$ split analysis}
In this subsection, we give $(\bar{D},q)$-RIP guarantee results on compressed data separation from noisy measurements $y=A(f_1+f_2+ \cdots + f_{\iota}) + z$ via solving
the following $l_q$-analysis optimization \be\label{LQSAnalysis}
\argmin_{\tilde{f}\in\mR^{\iota n}} \|\Psi^*\tilde{f}\|_q\quad \mbox{subject
to}\quad \|A\bar{D}\Psi^*\tilde{f}-y\|_r \leq \varepsilon, \ee
where $0<q\leq 1 \leq r\leq \infty, \varepsilon\geq 0$, $\|z\|_r\leq \varepsilon$
and $\bar{D},\Psi,f$ are as in (\ref{e-notation}).
\begin{thm}\label{thm4}
Let $\iota\in \mathbb{Z}^+$, $0<q\leq 1 \leq r\leq \infty, \varepsilon\geq 0$.
Suppose we observe data from the model $y=A(f_1+f_2+\cdots+f_{\iota}) + z$ with $\|z\|_r\leq \varepsilon.$
Let $D_1\in\mR^{n\times d_1}$,$D_2\in\mR^{n\times d_2}$,$\cdots, D_{\iota} \in \mR^{n \times d_{\iota}},$ be $\iota$ arbitrary tight frames for $\mR^n$ with frame bounds $1$. Let
$\bar{D},\Psi,f$ be as in (\ref{e-notation}). Fix a positive integer
$a>s$. Assume that the mutual coherence $\mu_1$ between $D_1
, D_2,\cdots,D_{\iota}$ satisfies \be\label{e_MIPC}\mu_1(s+a)(\rho^{2/q-1}+1)<1,\ee
and that the $(\bar{D},q)$-RIP constant of $A$ satisfies
\begin{equation}\label{e-cd}
\Delta\rho^{1-q/2}(\rho^{2/q-1}+1)^{q/2}<(2\iota)^{-q/2} , \end{equation}
where \be\label{rhodelta}
\rho=\frac{s}{a}\quad \mbox{ and }\quad
\Delta=\frac{1+\delta_a}{1-\delta_{s+a}}.\ee
 Then any solution $\hat{f}$ to the $l_q$
Split analysis (\ref{LQSAnalysis}) obeys
\be\label{R2}\|\hat{f}-f\|_2\leq
\tilde{C}_1\frac{\|\Psi^*f-(\Psi^*f)_{[s]}\|_q}{s^{1/q-1/2}} +
\tilde{C}_2m^{1/q-1/r} \varepsilon,\ee and \be \label{EstimationDh2}
\|\Psi^*\hat{f} - \Psi^*f\|_q^q \leq
\tilde{C}_3\|\Psi^*f-(\Psi^*f)_{[s]}\|_q^q +
\tilde{C}_4a^{1-q/2}m^{1-q/r}\varepsilon^q,\ee
   where $\tilde{C}_1,\tilde{C}_2,\tilde{C}_3$ and $\tilde{C}_4$ are
  positive constants (given explicitly in the proof) depending only on the
  $(\bar{D},q)$-RIP constant $\delta_s$, $\delta_{s+a}$, $\rho,\iota$ and $q$.
\end{thm}

The proof of this theorem will be given in Subsection 3.3. The following result is a direct
consequence of the above theorem. We will use it to prove Theorem \ref{thm3}.
\begin{corollary}\label{Cor2}
  Under the assumptions of Theorem \ref{thm4}, we further assume that $z=0$.
  Then any solution $\hat{f}$ to (\ref{SAAI}) satisfies
  $$ \|\hat{f}-f\|_2\leq
  \tilde{C}_1\frac{\|\Psi^*f-(\Psi^*f)_{[s]}\|_q}{s^{1/q-1/2}} ,$$
  where $\tilde{C}_1$ is given by Theorem \ref{thm4}.
\end{corollary}

\subsection{Proof of Theorem \ref{thm3}}
Now, we are ready to prove Theorem \ref{thm3}.
The rest of the proof is similar to the argument used in the proof
of Theorem \ref{thm2}. We include the sketch only. We first prove
that (\ref{e-cd}) holds with high probability. Setting $\iota=2$ in Theorem \ref{thm4}, it suffices to prove
$$\Delta\rho^{1-q/2}<2^{-3q/2}.$$ Let
$k=a+s=(t+1)s=(\lceil(2^{3q/2}b)^{\frac{2}{2-q}}\rceil+1)s,$ where
  $b>1.$ We only need to prove that
  $\delta_{(t+1)s}<\frac{b-1}{b+1}.$ A similar argument as that in
  the proof of Theorem \ref{thm2}, one can easily prove that
  $\delta_{(t+1)s}<\frac{b-1}{b+1}$ is met with probability
  exceeding $1/\binom{d}{s}$ provided that \bea m&\geq&
6.25q\beta_q^2\left[
\left(\lceil(5\cdot2^{3q/2})^{\frac{2}{2-q}}\rceil+1\right)\left(\ln
3-\ln\left(\lceil(5\cdot2^{3q/2})^{\frac{2}{2-q}}\rceil+1\right)\right)s+\ln
2\right.\\
&&\left.+\left(\lceil(5\cdot2^{3q/2})^{\frac{2}{2-q}}\rceil+2\right)s\ln
\left(\frac{\mathrm{e}d}{s}\right)\right]
+17.6\beta_q^2(\lceil(5\cdot2^{3q/2})^{\frac{2}{2-q}}\rceil+1)s
.\eea Note that in the proof we set $b=5$. In this case,
(\ref{e_MIPC}) is implied by
$$\mu_1s\left(\lceil(2^{3q/2}5)^{\frac{2}{2-q}}\rceil+1\right)\left(\frac{1}{8\cdot 5^{2/q}}+1\right)<1.$$
Now one can finish the proof by applying Corollary \ref{Cor2} and
\bea\|\Psi^*f-(\Psi^*f)_{[s]}\|_q^q&\leq&
\|D_1^*f_1-(D_1^*f_1)_{[s_1]}\|_q^q+\|D_2^*f_2-(D_2^*f_2)_{[s_2]}\|_q^q.\eea

\subsection{Proof of Theorem \ref{thm4}}
In this subsection, we prove Theorem \ref{thm4}.
 Our goal is to bound the norm
of $h$, where $h=f-\hat{f}$, $\hat{f}$ is a
solution of (\ref{SAAI}) and $f$ is the original signal. As in the proof of Theorem \ref{thm1},
 We do so
by bounding the norm of $\Psi^*h$, since $\Psi$ is a tight frame for $\mR^{\iota n}$. Actually,
since $D_1,D_2,\cdots,D_{\iota}$ are tight frames with frame bounds $1$, one has that
\be\label{e-psi}\|\Psi^*\tilde{f}\|_2^2=\sum_{k=1}^{\iota}\|D_k^*\tilde{f}_k\|_2^2=
\sum_{k=1}^{\iota}\|\tilde{f}_k\|_2^2=\|\tilde{f}\|_2^2\quad\mbox{
for all } \tilde{f}\in\mR^{\iota n},\ee
and that
\be\label{e-d}\|\bar{D}\|=\sqrt{\lambda_{\max}(\bar{D}\bar{D}^*)}=\sqrt{\lambda_{\max}(\iota I)}=\sqrt{\iota}.\ee

For arbitrary fixed $h\in \mR^{\iota n},$ we write
$\Psi^*h=(x_1,x_2,...,x_{\bar{d}})^T$. Making rearrangements if necessary,
we assume that $|x_{1}|\geq|x_{2}|\geq\cdots|x_{\bar{d}}|.$ Set
$T=T_{0}=\{1,2,\cdots,s\},$ $T_{1}=\{s+1,s+2,\cdots,s+a\},$ and
$T_{i}=\{s+(i-1)a+1,\cdots,s+ia\},i=2,\cdots,$ with the last subset
of size less than or equal to $a$. Denote $T_{01} = T_0\cup T_1.$
Note that by applying the first inequality of
(\ref{BasicInequality}), we have \be\label{SFinalEstimate1}
\|h\|_2^2 = \|\Psi^*h\|_2^2
\leq\|\Psi_{T_{01}}^*h\|_2^2+\left(\sum_{j\geq2}\|\Psi_{T_{j}}^*h\|_2^q\right)^{2/q}.
\ee We only need to upper bound the last two terms. By Lemma
\ref{TailEstimation}, we also have
\be\label{e_sum2}\sum_{j\geq2}\|\Psi_{T_j}^*h\|_2^q=
a^{q/2-1}\|\Psi_{T^c}^*h\|_q^q. \ee Applying (\ref{e_sum2}) to
(\ref{SFinalEstimate1}), we can upper bound $\|h\|_2^2$ by
\be\label{SFinalEstimate2} \|h\|_2^2 \leq
\|\Psi_{T_{01}}^*h\|_2^2+a^{1-2/q}\|\Psi_{T^c}^*h\|_q^2.\ee To bound
$\|\Psi_{T_{01}}^*h\|_2^q,$ we need the following lemma, which gives
an upper bound on
$\sum_{k=1}^{\iota}\|D_kD_{kT_{01}^k}^*h_1\|_2^2$.
 It can be proved by using of the definitions of the $(\bar{D},q)$-RIP and the mutual coherence.
 Denote that $T_{01}^1 = T_{01} \cap [d_1]$ and $T_{01}^k = \{j-\sum_{l=1}^{k-1}d_l : j \in T_{01} \cap [\sum_{l=1}^{k-1}d_l+1,\sum_{l=1}^{k}d_l]\} $ for any integer $k\in [2,\iota].$
\begin{lemma}[Consequence of the $(\bar{D},q)$-RIP and the mutual coherence]\label{lemDRIPMU}
 Assume that $A$
 satisfies the $(\bar{D},q)$-RIP of order $s+a.$
Let $\Delta$ be as in \ref{rhodelta} and
\be\label{clustercoherence} U=\frac{\mu_1(s+a)}{2},\ee
 where $u_1$ is the mutual coherence between $D_1$, $D_2,\cdots,D_{\iota}$.
Then we have
\be\label{eDDT12}
\sum_{k=1}^{\iota}\|D_kD_{kT_{01}^k}^*h_k\|_2^2 \leq \iota \Delta^{2/q} a^{1-2/q} \left( \|\Psi_{T^c}^*h\|_q^q + { a^{1 - q/2} \|A\bar{D}\Psi^*h\|_q^q \over \iota^{q/2}(1+\delta_{a})} \right)^{2/q} + U\|\Psi_{T_{01}}^*h\|_2^2.
\ee
\end{lemma}
\begin{proof}
We first note that from the definition of
$\mu_1$, \bea &&\left|\sum_{k=1}^{\iota} \sum_{l\neq k}\la
D_kD_{kT_{01}^k}^*h_k,D_lD_{lT_{01}^l}^*h_l\ra\right| =2\left|\sum_{k=1}^{\iota} \sum_{l\neq k} \sum_{i\in T_{01}^k }\sum_{j\in T_{01}^l } \la
d_{ki},d_{lj}\ra \la
d_{ki},h_k\ra\la h_l,d_{lj}\ra\right|\\
&\leq&\mu_1\sum_{k=1}^{\iota} \sum_{l\neq k} \sum_{i\in T_{01}^k }\sum_{j\in T_{01}^l }\left| \la d_{ki},h_k\ra\la h_l,d_{lj}\ra\right|
=\mu_1\sum_{k=1}^{\iota} \sum_{l\neq k} \|D_{kT_{01}^k}^*h_k\|_1\|D_{lT_{01}^l}^*h_l\|_1.\eea
Using Cauchy-Scwarz inequality, for $k,l\in [\iota],$ $k\neq l$, we have
\bea
&&\|D_{kT_{01}^k}^*h_k\|_1\|D_{lT_{01}^l}^*h_l\|_1 \leq \sqrt{|T_{01}^k| \cdot |T_{01}^l|}\|D_{kT_{01}^k}^*h_k\|_2 \|D_{lT_{01}^l}^*h_l\|_2 \\
&\leq& {|T_{01}^k| + |T_{01}^l| \over 2} \|D_{kT_{01}^k}^*h_k\|_2 \|D_{lT_{01}^l}^*h_l\|_2 \leq {s+a \over 2} \|D_{kT_{01}^k}^*h_k\|_2 \|D_{lT_{01}^l}^*h_l\|_2.
\eea
Therefore, recalling that $U$ is given by (\ref{clustercoherence}), we get
\bea \left|\sum_{k=1}^{\iota} \sum_{l\neq k}\la
D_kD_{kT_{01}^k}^*h_k,D_lD_{lT_{01}^l}^*h_l\ra\right|
&\leq&  U\sum_{k=1}^{\iota} \sum_{l\neq k} \|D_{kT_{01}^k}^*h_k\|_2\|D_{lT_{01}^l}^*h_l\|_2 \\
&\leq& U\sum_{k=1}^{\iota} \|D_{kT_{01}^k}^*h_k\|_2^2 = U\|\Psi_{T_{01}}^*h\|_2^2.
\eea
Thus, we get
\begin{eqnarray}\|\bar{D}\Psi_{T_{01}}^*h\|_2^2&=&\left\|\sum_{k=1}^{\iota} D_kD_{kT_{01}^k}^*h_k\right\|_2^2\nonumber\\
&=&\sum_{k=1}^{\iota}\|D_kD_{kT_{01}^k}^*h_k\|_2^2 + \sum_{k=1}^{\iota} \sum_{l\neq k}\la
D_kD_{kT_{01}^k}^*h_k,D_lD_{lT_{01}^l}^*h_l\ra\nonumber\\
\label{e_MIP}&\geq&\sum_{k=1}^{\iota}\|D_kD_{kT_{01}^k}^*h_k\|_2^2
-U\|\Psi_{T_{01}}^*h\|_2^2.\end{eqnarray}
We next upper bound $\|\bar{D}\Psi_{T_{01}}^*h\|_2^2$. We do so by using properties of $(\bar{D},q)$-RIP.
Note that \bea \|A\bar{D}\Psi^*h\|_q^q = \left\|A\bar{D}\Psi_{T_{01}}^*h+\sum_{j\geq2}A\bar{D}\Psi_{T_j}^*h\right\|_q^q
\geq \|A\bar{D}\Psi_{T_{01}}^*h\|_q^q - \sum_{j\geq2}\|A\bar{D}\Psi_{T_j}^*h\|_q^q.\eea
According to the definition of $(\bar{D},q)$-RIP, and then applying
(\ref{e-d}),
\bea\|A\bar{D}\Psi^*h\|_q^q &\geq&
(1-\delta_{s+a})\|\bar{D}\Psi_{T_{01}}^*h\|_2^q-
(1+\delta_{a})\sum_{j\geq2}\|\bar{D}\Psi_{T_j}^*h\|_2^q\\
&\geq&(1-\delta_{s+a})\|\bar{D}\Psi_{T_{01}}^*h\|_2^q - \iota^{q/2}(1+\delta_{a})\sum_{j\geq2}\|\Psi_{T_j}^*h\|_2^q.\eea
Introducing (\ref{e_sum2}) to the above and then dividing both sides by $1-\delta_{s+a}$, with (\ref{rhodelta}),  we get
\bea\|\bar{D}\Psi_{T_{01}}^*h\|_2^q \leq \iota^{q/2}\Delta
a^{q/2-1}\|\Psi_{T^c}^*h\|_q^q + \|A\bar{D}\Psi^*h\|_q^q/(1-\delta_{s+a}),\eea which is equivalent to
\be\label{e4}\|\bar{D}\Psi_{T_{01}}^*h\|_2^2\leq
\iota \Delta^{2/q} a^{1-2/q} \left( \|\Psi_{T^c}^*h\|_q^q + { a^{1 - q/2} \|A\bar{D}\Psi^*h\|_q^q \over \iota^{q/2}(1+\delta_{a})} \right)^{2/q}.\ee
Introducing (\ref{e_MIP}) to (\ref{e4}), one can get the desired result.
\end{proof}
Now, we give an upper bound on $\|\Psi_{T_{01}}^*h\|_2^q.$ By developing a relationship between
$\|\Psi_{T_{01}}^*h\|_2^2$ and $\sum_{k=1}^{\iota}\|D_kD_{kT_{01}^k}^*h_k\|_2^2,$
and then applying the above lemma, one can prove the following result.
\begin{lemma}
Under the assumptions of Lemma \ref{lemDRIPMU}, assume that $U<1.$ Then, we have
\be\label{PsiT01}
\|\Psi_{T_{01}}^*h\|_2^q \leq  \left({ U + \iota\Delta^{2/q} + \sqrt{(U - \iota \Delta^{2/q} )^2 + 4\iota \Delta^{2/q} } \over 2(1-U)}\right)^{q/2}{1\over a ^{1-q/2}} \left( \|\Psi_{T^c}^*h\|_q^q + {a^{1 - q/2} \|A\bar{D}\Psi^*h\|_q^q \over \iota^{q/2}(1+\delta_{a})} \right).
\ee
\end{lemma}
\begin{proof}
  Note that by applying Cauchy-Schwarz inequality twice,
  \bea
&& \|\Psi_{T_{01}}^*h\|_2^4 = \left(\sum_{k=1}^{\iota} \|D_{kT_{01}^k}^*h_k\|_2^2\right)^2 = \left(\sum_{k=1}^{\iota} \la D_kD_{kT_{01}^k}^*h_k,h_k\ra\right)^2\\
&\leq& \left(\sum_{k=1}^{\iota}\|h_k\|_2 \| D_kD_{kT_{01}^k}^*h_k\|_2\right)^2 \leq \left(\sum_{k=1}^{\iota}\|h_k\|_2^2\right)\left(\sum_{k=1}^{\iota}\| D_kD_{kT_{01}^k}^*h_k\|_2^2\right)\\
& \leq & \|h\|_2^2 \left(\sum_{k=1}^{\iota}\| D_kD_{kT_{01}^k}^*h_k\|_2^2\right).
\eea
Applying (\ref{SFinalEstimate2}), we get
\bea
\|\Psi_{T_{01}}^*h\|_2^4 \leq \left(\|\Psi_{T_{01}}^*h\|_2^2+a^{1-2/q}\|\Psi_{T^c}^*h\|_q^2\right)\left(\sum_{k=1}^{\iota}\| D_kD_{kT_{01}^k}^*h_k\|_2^2\right).
\eea
Introducing (\ref{eDDT12}),
\bea
\|\Psi_{T_{01}}^*h\|_2^4 \leq \left(\|\Psi_{T_{01}}^*h\|_2^2+a^{1-2/q}\|\Psi_{T^c}^*h\|_q^2\right)
\left(\iota\Delta^{2/q}a ^{1-2/q} \mathcal{X} + U\|\Psi_{T_{01}}^*h\|_2^2\right),
\eea
where for notational simplicity, we set
\be\label{mathcalX}\mathcal{X} = \left(\|\Psi_{T^c}^*h\|_q^q + {a^{1 - q/2} \|AD\Psi^*h\|_q^q \over \iota^{q/2}(1+\delta_{a})} \right)^{2/q}.\ee
Rearranging terms, this can be rewritten as
$$ (1-U) \|\Psi_{T_{01}}^*h\|_2^4 - (U\|\Psi_{T^c}^*h\|_q^2 + \iota \Delta^{2/q} \mathcal{X})a^{1-2/q}\|\Psi_{T_{01}}^*h\|_2^2 - \iota a^{2-4/q}\|\Psi_{T^c}^*h\|_q^2 \Delta^{2/q} \mathcal{X} \leq 0  .$$
Noting that $U<1$ and by solving a quadratic inequalities of type of $c_1x^2 -c_2 x - c_3 \leq 0 $ with variable $x\in[0,\infty)$ and positive constants $c_1,c_2,c_3$,
we get $$\|\Psi_{T_{01}}^*h\|_2^2 \leq { (U\|\Psi_{T^c}^*h\|_q^2 + \iota \Delta^{2/q} \mathcal{X}) + \sqrt{(U \|\Psi_{T^c}^*h\|_q^2 + \iota \Delta^{2/q} \mathcal{X})^2 + 4\iota (1-U)\|\Psi_{T^c}^*h\|_q^2 \Delta^{2/q} \mathcal{X}} \over 2(1-U)a^{2/q-1}}.$$
Upper bounding the term $\|\Psi_{T^c}^*h\|_q^2$ by $\mathcal{X},$
\bea
 \|\Psi_{T_{01}}^*h\|_2^2 &\leq&  { (U + \iota \Delta^{2/q}) \mathcal{X} + \sqrt{(U + \iota \Delta^{2/q} )^2\mathcal{X}^2 + 4\iota (1-U)\Delta^{2/q} \mathcal{X}^2} \over 2(1-U)a^{2/q-1}}\\
 &=& { U + \iota \Delta^{2/q} + \sqrt{(U - \iota \Delta^{2/q} )^2 + 4\iota \Delta^{2/q} } \over 2(1-U)a^{2/q-1}}\mathcal{X},
 \eea
which leads to our desired result by introducing (\ref{mathcalX}) and taking the $(q/2)$-th power on both sides.
\end{proof}
In what follows, we shall bound $\|\Psi_{T^c}^*h\|_q^q$ and $\|A\bar{D}\Psi^*h\|_q^q.$
We first need the following result, which shows that the $(\bar{D},q)$-RIP implies
that the matrix $A\bar{D}\Psi^*$ satisfies robust $\Psi$-NSP$_q$.
\begin{lemma}[Robust $\Psi$-NSP$_q$]\label{LemNSP2}
Under the assumptions of Lemma \ref{lemDRIPMU}, we have
\be\label{e_DNSP2}
\|\Psi_{T}^*h\|_q^q \leq \tilde{\theta}\left(\|\Psi_{T^c}^*h\|_q^q + {a^{1 - q/2} \|A\bar{D}\Psi^*h\|_q^q \over \iota^{q/2}(1+\delta_{a})}\right),\ee
 where
\be\label{tildetheta}\tilde{\theta}=\left({ U + \iota \Delta^{2/q} + \sqrt{(U - \iota \Delta^{2/q} )^2 + 4\iota \Delta^{2/q} } \over 2(1-U)}\right)^{q/2} \rho ^{1 - q/2}.\ee
In particular, if
\be\label{DRipMipCondition} \iota \Delta^{2/q} ( \rho ^{ 2/q -1} + 1)\rho^{2/q -1} + U(1 +  \rho ^{2/q - 1}) < 1, \ee
 then $\tilde{\theta}<1.$
\end{lemma}
\begin{proof}
  Note that by applying (\ref{BasicInequality}), one has \bea
\|\Psi_{T_{01}}^*h\|_2^q \geq\|\Psi_{T}^*h\|_2^q\geq
\|\Psi_{T}^*h\|_q^q/s^{1-q/2}.\eea
 Combining with (\ref{PsiT01}), we get
 \bea
{\|\Psi_{T}^*h\|_q^q \over s^{1-q/2}}  \leq
  \left({ U + \iota \Delta^{2/q} + \sqrt{(U - \iota \Delta^{2/q} )^2 + 4\iota \Delta^{2/q} } \over 2(1-U)}\right)^{q/2}{1\over a ^{1-q/2}} \left( \|\Psi_{T^c}^*h\|_q^q + {a^{1 - q/2} \|A\bar{D}\Psi^*h\|_q^q \over \iota^{q/2}(1+\delta_{a})} \right),
\eea
which is equivalent to (\ref{e_DNSP2}).

It remains to prove $\tilde{\theta}<1$. This can be verified by showing that $\tilde{\theta}^{2/q}<1,$
which is guaranteed (since $U<1$ by our assumptions) provided that
$$ \sqrt{(U - \iota\Delta^{2/q} )^2 + 4\iota\Delta^{2/q} } < 2 \rho ^{1 - 2/q}(1-U) - (U + \iota \Delta^{2/q}). $$
Note that under assumption (\ref{DRipMipCondition}), the right hand
side is always positive. Taking the second power on both sides,
rearranging terms, a sufficient condition for the above is
$$ (U - \iota \Delta^{2/q} )^2 + 4\iota \Delta^{2/q} - (U + \iota \Delta^{2/q})^2 <  4 \rho ^{2 - 4/q}(1-U)^2 - 4 \rho ^{1 - 2/q}(1-U) (U + \iota \Delta^{2/q}),$$
which can be rewritten as
$$ 4\iota \Delta^{2/q}(1-U) < 4 \rho ^{2 - 4/q}(1-U)[ 1 - U( 1 +  \rho ^{2/q - 1})] - 4\iota \rho ^{1 - 2/q}(1-U) \Delta^{2/q}. $$
Dividing both sides by $4(1-U)\rho ^{2 - 4/q}$ and rearranging terms, a sufficient condition for the above is (\ref{DRipMipCondition}).
From the above analysis, we have $\tilde{\theta}<1$. The proof is finished.
\end{proof}

Now, we shall choose $h=\hat{f}-f$, where $\hat{f}$ is a solution of $(\ref{LQSAnalysis})$ and
$f$ is the ``original signal" given by (\ref{LQSAnalysis}). Let ${\Omega}$ be the index set of the $s$ largest entries of $|\Psi^*f|.$ The following results can be verified by using the fact that $\hat{f}$ is a solution of $(\ref{LQAnalysis})$.
By a similar argument as that for Lemma \ref{LemSolution}, one
gets that \be\label{e_cone constraint2}
\|\Psi^*_{T^c}h\|_q^q\leq\|\Psi^*_{T}h\|_q^q+2\|\Psi^*_{\Omega^c}f\|_q^q\ee
and \be\label{errorbound2} \|A\bar{D}\Psi^*h\|_q^q\leq  m^{1-q/r} (2\varepsilon)^q. \ee

We may now conclude the proof of Theorem \ref{thm4}. Note that assumptions (\ref{e_MIPC}) and (\ref{e-cd}) imply (\ref{DRipMipCondition}) since $U$ is given by (\ref{clustercoherence}).
Therefore, by Lemma \ref{LemNSP2}, we have $\tilde{\theta} < 1.$
Combining (\ref{e_cone constraint2}) with (\ref{e_DNSP2}),
\be\label{e-3}\|\Psi^*_{T}h\|_q^q\leq\frac{\tilde{\theta}}{1-\tilde{\theta}} \left(2\|\Psi^*_{\Omega^c}f\|_q^q +  {a^{1 - q/2} \|A\bar{D}\Psi^*h\|_q^q \over \iota^{q/2}(1+\delta_{a})}\right),\ee
and
\be\label{e-2}\|\Psi^*_{T^c}h\|_q^q\leq\frac{1}{1-\tilde{\theta}} \left(2\|\Psi^*_{\Omega^c}f\|_q^q + {  \tilde{\theta}a^{1 - q/2} \|A\bar{D}\Psi^*h\|_q^q \over \iota^{q/2}(1+\delta_{a})}\right).\ee
By (\ref{SFinalEstimate2}), (\ref{PsiT01}) and noting that $\tilde{\theta}$ is given by (\ref{tildetheta}), and $\rho$ is given by (\ref{rhodelta}),
\bea\|h\|_2^2&\leq &\|\Psi_{T_{01}}^*h\|_2^2+a^{1-2/q} \|\Psi_{T^c}^*h\|_q^2\\
&\leq& \theta^{2/q}s^{1-2/q}\left(\|\Psi_{T^c}^*h\|_q^q + {a^{1-q/2} \|A \bar{D} \Psi^*h\|_q^q \over \iota^{q/2} (1 + \delta_a)}\right)^{2/q}
+\rho^{2/q-1}s^{1-2/q}\|\Psi_{T^c}^*h\|_q^2 \\
&\leq&{\theta^{2/q} + \rho^{2/q-1} \over s^{2/q -1} }  \left(\|\Psi_{T^c}^*h\|_q^q + {a^{1-q/2} \|A \bar{D} \Psi^*h\|_q^q \over \iota^{q/2} (1 + \delta_a)}\right)^{2/q}.\eea
Introducing (\ref{e-2}), and then applying (\ref{errorbound2}) and a basic inequality $(b+c)^{1/q}\leq 2^{1/q-1}(b^{1/q}+c^{1/q}),\forall b,c\geq 0$, with $\rho$ given by (\ref{rhodelta}),
\bea \|h\|_2^2
&\leq&\frac{\tilde{\theta}^{2/q}+\rho^{2/q-1}}{(1-\tilde{\theta})^{2/q}s^{2/q-1}}\left(2\|\Psi^*_{\Omega^c}f\|_q^q + { (1 + \tilde{\theta})a^{1 - q/2} \|A\bar{D}\Psi^*h\|_q^q \over \iota^{q/2}(1+\delta_{a})}\right)^{2/q}\\
&\leq&  \frac{\tilde{\theta}^{2/q}+\rho^{2/q-1}}{(1- \tilde{\theta})^{2/q}s^{2/q-1}}\left(2\|\Psi^*_{\Omega^c}f\|_q^q + { (1+\tilde{\theta}) a^{1 - q/2} m^{1-q/r} (2\varepsilon)^q \over  \iota^{q/2}(1+\delta_{a})}\right)^{2/q}\\
&\leq& \frac{(\tilde{\theta}^{2/q}+\rho^{2/q-1})}{(1-\tilde{\theta})^{2/q}} \left( {2^{2/q-1}\|\Psi^*_{\Omega^c}f\|_q \over s^{1/q-1/2}} + { 2^{1/ q} (1+\tilde{\theta})^{1/q} \rho^{1/2-1/q } m^{1/q-1/r} \varepsilon \over \sqrt{\iota}(1+\delta_{a})^{1/q}}\right)^2.\eea
Taking the square power on both sides, we get
\bea \|h\|_2 \leq \frac{2^{2/q-1}(\tilde{\theta}^{1/q}+\rho^{1/q-1/2})}{(1-\tilde{\theta})^{1/q}} {\|\Psi^*_{\Omega^c}f\|_q \over s^{1/q-1/2}}
 + \frac{2^{1/q} \rho^{1/2-1/q } (1+ \tilde{\theta})^{1/ q}(\tilde{\theta}^{1/q}+\rho^{1/q-1/2})}{\sqrt{\iota}(1-\tilde{\theta})^{1/q}(1+\delta_{a})^{1/q}}  m^{1/q-1/r} \varepsilon, \eea
which leads to the desired result (\ref{R2}).
To prove (\ref{EstimationDh2}), we first apply (\ref{e-3}), (\ref{e-2}) and then use (\ref{errorbound2}) to get
\bea \|\Psi^*h\|_q^q &=& \|\Psi_T^*h\|_q^q + \|\Psi_{T^c}^*h\|_q^q \\
&\leq& \frac{2(1+\tilde{\theta})}{1-\tilde{\theta}} \|\Psi^*_{\Omega^c}f\|_q^q +
{2 \tilde{\theta} \over \iota^{q/2}(1 - \tilde{\theta})(1+\delta_{a})} a^{1 - q/2} \|A\bar{D}\Psi^*h\|_q^q\\
&\leq& \frac{2(1+\tilde{\theta})}{1-\tilde{\theta}} \|\Psi^*_{\Omega^c}f\|_q^q +
{2^{q+1} \tilde{\theta}  \over \iota^{q/2} (1 - \tilde{\theta})(1+\delta_{a})} a^{1 - q/2} m^{1-q/r} \varepsilon^q, \eea
which leads to the desired result (\ref{EstimationDh2}). The proof is finished.
\begin{rem}
\begin{enumerate}
  \item  In the proof, we have proved that suitable conditions on the modified $q$-RIP and the mutual coherence
  between the dictionaries $D_1,D_2,\cdots,D_{\iota}$ imply that $A\bar{D}\Psi$
  satisfies $\Psi$-NSP$_q$, i.e., $$
\|\Psi_{T}^*h\|_q^q\leq \tilde{\theta}\|\Psi_{T^c}^*h\|_q^q,$$ for all $h\in
\ker (A\bar{D}\Psi)$ and all $T\subset[d]$ with $|T|\leq s$, where
$\tilde{\theta}<1.$ Now we can define the null space property for compressed
data separation as follow: For $q\in(0,1]$, $A$ is said to satisfy
the $l_q$-split null space property with respect to dictionaries $D_1$
 $D_2,\cdots,D_{\iota}$ of order $s$ if there exists some $\theta \in [0,1)$ such that
$$
\sum_{k=1}^{\iota} \|D_{kT_k}^*h_k\|_q^q \leq\theta \sum_{k=1}^{\iota} \|D_{kT_k^c}^*h_k\|_q^q,$$
for all $h_1,h_2,\cdots,h_{\iota}$ such that $(\sum_{k=1}^{\iota}h_k) \in\ker(A)$ and all
$T_1\subset[d_1]$, $T_2\subset[d_2],\cdots,T_{\iota} \in [d_{\iota}]$ with $\sum_{k=1}^{\iota}|T_k|\leq s$. Note that the null space property for standard
compressed sensing is one of the well known conditions on
measurement matrices (e.g. \cite{GN07,CDD09,CZ12}). Here, we provide
a definition of null space property for compressed data separation,
which may be of independent interest.
  \item From the proof, we see that the following inequality holds $$\|\Psi^*\hat{f}-\Psi^*f\|_2\leq
C_1\frac{\|\Psi^*f-(\Psi^*f)_{[s]}\|_q}{s^{1/q-1/2}} + C_2 m^{1/q-1/r} \varepsilon .$$
\end{enumerate}
\end{rem}

\section{Numerical realization and discussion}
In this final section, we discuss numerical realization of the constrained $l_q$ analysis $(P_q)$,
and provide further discussions on our theoretical analysis.

The constrained $l_q$ analysis problem $(P_q)$ proposed to recover $f$ is nonconvex.
Due to its nonconvexity, finding a global minimizer of problem $(P_q)$ is generally NP-hard.
We thus solve such a nonconvex problem by solving a sequence of
convex problems, as often done in standard compressed sensing for standard $l_q$ minimization, e.g., \cite{FL09,DDFG10}. The first possible method is the iteratively reweighted $l_1$ analysis which iteratively solves the following weighted $l_1$ analysis (IRL1)
$$
f^{j+1} = \argmin_{\tilde{f} \in \mR^n} \sum_{i=1}^d \omega_i^j |\la d_i,\tilde{f}\ra| \qquad \mbox{subject to}\quad A\tilde{f} = y,
$$
where $d_i$ is the $i$-th column of $D$, $\omega_i^j = (|\la d_i,f^j\ra| + \varsigma_j)^{q-1}$ and $(\varsigma_j)$ a nonincreasing sequence of positive numbers. Another potential method is iteratively
reweighted least squares (IRLS), which iteratively solves the following
weighted least square problem
$$
f^{j+1} = \argmin_{\tilde{f} \in \mR^n} \sum_{i=1}^d \omega_i^j |\la d_i,\tilde{f}\ra|^2 \qquad \mbox{subject to}\quad A\tilde{f} = y,
$$
where $\omega_i^j = (|\la d_i,f^j\ra|^2 + \varsigma_j)^{q/2-1}$.
The key of the above methods lie in that the object functions in both methods are approximations to $\|D^*\tilde{f}\|_q^q$ when  $\varsigma \to 0$. In fact, one may derive convergence results for both these methods using the techniques developed in this paper and in \cite{FL09,DDFG10}. We postpone the details for a future work. We here provide a simple simulation, demonstrating that the analysis IRLS can solve the constrained $l_q$ analysis $(P_q)$ efficiently. In this simulation, we let $n=100,$ $d=110$, $m = 50$, $q=0.7$ and the sparsity of $D^*f$ as $s = 25$.
The entries in the $m \times n$ measurement matrix $A$ were randomly generated according to a normal distribution. The $n \times d$ matrix $d$ is a random tight frame, generated by the approach in \cite{NDEG13}. Figure 1 shows that the analysis IRLS reconstructs the signal $f$ exactly.

\begin{figure}
    \centering
  \includegraphics[width=1.0\textwidth]{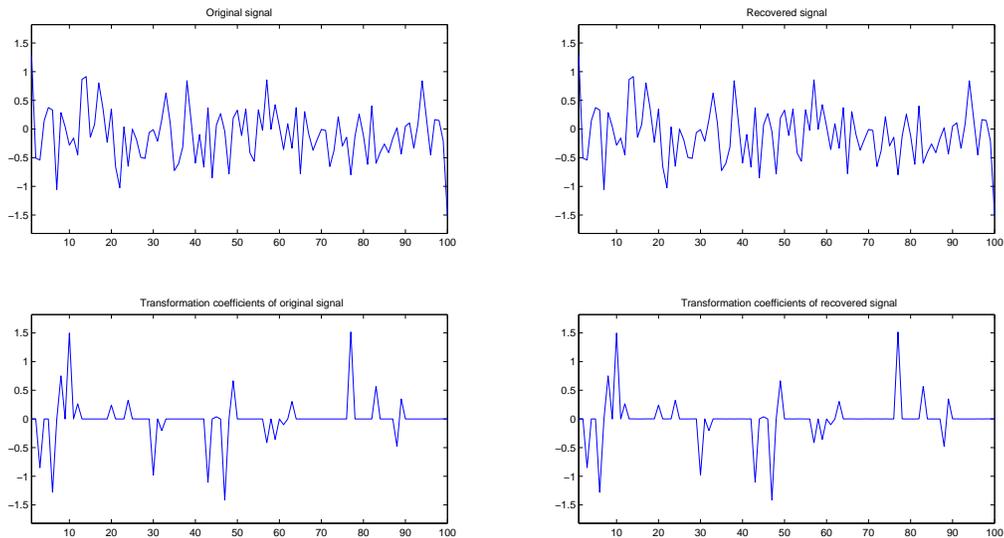}
   \caption{Reconstruction via analysis IRLS}
\end{figure}

In this paper, we discussed sparse recovery with general frames
from random measurements via the $l_q$-analysis optimization ($P_q$) with
$0<q\leq 1$. We introduced a notion of $(D,q)$-RIP. It is a natural extension of
the standard $q$-RIP defined by Chartrand and Staneva
\cite{CS08} for standard compressed sensing,
and is different from the $D$-RIP defined in (\ref{StandardDRIP}).
We established an $(D^{\dag},q)$-RIP guarantee result for the $l_q$-analysis
optimization ($P_q$). We proved the result by investigating the relationship
between the $(D^{\dag},q)$-RIP constant and the $D$-NSP$_q$ constant, which at the same time
may shed some lights on how to establish a tighter relationship
between the $D^{\dag}$-RIP constant and the $D$-NSP$_q$ constant.
Subsequently, we showed how many random Gaussian
measurements are needed for the $(D^{\dag},q)$-RIP condition to hold with high
probability. Finally, we discussed compressed data
separation by using the introduced $(D,q)$-RIP.
We showed that under an usual assumption that the two
dictionaries satisfy a mutual coherence condition, the $l_q$ split
analysis with $0<q\leq1 $ can approximately reconstruct the distinct
components from fewer random Gaussian measurements with small $q$
than when $q=1$. Our results provide theoretical basis for further designing
algorithm to solve ($P_q$) and the $l_q$ split
analysis optimization. Our proof techniques may shed some lights on improving the previous
$D$-RIP guarantee results. Further issues are to design
numerical methods (e.g., iteratively reweighted method) to find approximate solutions of the $l_q$-analysis optimization applied in practical applications,
and to consider other
random measurements instead of Gaussian measurements.


\begin{thebibliography}{99}
\setlength{\itemsep}{0.2pt}
\small
\bibitem{ACP11} A. Aldroubi, X. Chen and A. M. Powell, \emph{Perturbations of measurement matrices and dictionaries in
compressed sensing}, Appl. Comput. Harmon. Anal., \textbf{33}
(2012), 282-291.
\bibitem{BD09} T. Blumensath and M. Davies, \emph{Sampling theorems for signals from the
union of finite-dimensional linear subspaces}, IEEE Trans. Inf.
Theory, \textbf{55} (2009), 1872-1882.
\bibitem{BDDW08} R. Baraniuk, M. Davenport, R. DeVore and M. Wakin, \emph{A simple
proof of the restricted isometry property for random matrices},
Constr. Approx., \textbf{28} (2008), 253-263.

\bibitem{CCS10} J. Cai, R. Chan and Z. Shen, \emph{Simultaneous cartoon
and texture inpainting}, Inverse Probl. Imag., \textbf{4} (2010),
379-395.
\bibitem{CDOS12} J. Cai, B. Dong, S. Osher and Z. Shen,
\emph{Image restoration: total variation, wavelet frames, and
beyond}, J. Amer. Math. Soc., \textbf{25} (2012), 1033-1089.
\bibitem{CZ12} T. Cai and A. Zhang, \emph{Sharp RIP bound for sparse signal and low-rank matrix
recovery}, Appl. Comput. Harmon. Anal., \textbf{35} (2013), 74-93.

\bibitem{CD04}E. J. Cand\`{e}s and D. L. Donoho, \emph{New tight frames of curvelets and
optimal representations of objects with piecewise $C^2$
singularities}, Comm. Pure Appl. Math., \textbf{57} (2004), 219-266.
\bibitem{CENR11} E. J. Cand\`{e}s, Y. C. Eldar, D. Needell and P.
Randall, \emph{Compressed sensing with coherent and redundant
dictionaries,} Appl. Comput. Harmon. Anal., \textbf{31} (2011),
59-73.
\bibitem{CRT06b} E. J. Cand\`{e}s, J. Romberg and T. Tao, \emph{Stable signal recovery from
incomplete and inaccurate measurements}, Comm. Pure Appl. Math.,
\textbf{59} (2006), 1207-1223.
\bibitem{CT06} E. J. Cand\`{e}s and T. Tao,  \emph{Near optimal signal recovery from
random projections: Universal encoding strategies?}, IEEE Trans.
Inform. Theory, \textbf{52} (2006), 5406-5425.
\bibitem{C07} R. Chartrand, \emph{Exact reconstruction of sparse signals via nonconvex minimization}, IEEE Signal Process. Lett. \textbf{14} (2007), 707-710.
\bibitem{CS08} R. Chartrand and V. Staneva, \emph{Restricted isometry properties and nonconvex compressive sensing},
Inverse Probl., \textbf{24} (2008), 1-14.
\bibitem{CDD09} A. Cohen,
W. Dahmen and R. DeVore, \emph{Compressed sensing and best k-term
approximation}, J. Amer. Math. Soc., \textbf{22} (2009), 211-231.
\bibitem{DH02} I. Daubechies and B. Han, \emph{The canonical dual frame of a wavelet
frame}, Appl. Comput. Harmon. Anal. \textbf{12} 2002, 269-285.
\bibitem{DDFG10} I. Daubechies, R. Devore,
M. Fornasier and S. Gunturk, \emph{Iteratively reweighted least
squares minimization for sparse recovery}, Comm. Pure. Appl. Math.,
\textbf{13} (2010), 1-38.
\bibitem{D06b} D. L. Donoho, \emph{For most large underdetermined systems of linear
equations the minimal $l^1$ solution is also the sparsest solution},
Comm. Pure Appl. Math., \textbf{59} (2006), 797-829.
IEEE Trans. Inform. Theory, \textbf{59} (2013), 6820-6829.
\bibitem{DK13} D. L. Donoho and G. Kutyniok, \emph{Microlocal analysis of the
geometric separation problem}, Comm. Pure Appl. Math., \textbf{66}
(2013), 1-47.
\bibitem{EMR07} M. Elad, P. Milanfar and R. Rubinstein, \emph{Analysis versus synthesis in signal priors}, Inverse Probl. \textbf{23} (2007) , 947-968.
\bibitem{ESQD05} M. Elad, J. L. Starck, P. Querre and D. L. Donoho, \emph{Simultaneous cartoon
and texture image inpainting using morphological component analysis
(MCA)}, Appl. Comput. Harmon. Anal., \textbf{19} (2005), 340-358.
\bibitem{FS98} H. Feichtinger, T. Strohmer (Eds.), \emph{Gabor Analysis and
Algorithms}, Birkh\"{a}user, 1998.
\bibitem{F14} S. Foucart, \emph{Stability and robustness of $\ell_1$-minimizations with Weibull matrices and redundant dictionaries},
Linear Algebra Appl., \textbf{441} (2014), 4-21.
\bibitem {FL09} S. Foucart and M. J. Lai,
\emph{Sparsest solutions of underdetermined linear systems via $l_q$
minimization for $0 < q \leq1$}, Appl. Comput. Harmon. Anal.,
\textbf{26} (2009), 395-407.
\bibitem{GN03} R. Gribonval and M. Nielsen, \emph{Sparse decompositions in unions of bases},
IEEE Trans. Inform. Theory, \textbf{49} (2003), 3320-3325.
\bibitem{GN07} R. Gribonval and M. Nielsen, \emph{Highly sparse representations from
dictionaries are unique and independent of the sparseness measure}.
Appl. Comput. Harmon. Anal., \textbf{22} (2007), 335-355.
\bibitem{KR13} M. Kabanava and H. Rauhut, \emph{Analysis $l_1$-recovery with frames and Gaussian
measurements}, Arxiv, 2013.
\bibitem{KW11} F. Kramer and R. Ward, \emph{New and improved Johnson-Lindenstrauss
embeddings via the restricted isometry property}, SIAM J. Math.
Anal., \textbf{43} (2011), 1269-1281.
\bibitem{LL11} S. Li and J. Lin, \emph{Compressed sensing with coherent tight frame via $l_q$-minimization}, \textbf{8} (2014), 761-777.
\bibitem{LLS13TIT} J. Lin, S. Li and Y. Shen, \emph{Compressed data separation
with coherent dictionaries}, IEEE Trans. Inform. Theory, \textbf{59}
(2013), 4309-4315.
\bibitem{LML12} Y. Liu, T. Mi and S. Li, \emph{Compressed sensing with general frames via optimal-dual-based
$\ell_1$-analysis}, IEEE Trans. Inform. Theory, \textbf{58} (2012),
4201-4214.
\bibitem{LD08} Y. Lu and M. Do, \emph{A theory for sampling signals from a
union of subspaces}, IEEE Trans. Signal Process., \textbf{56}
(2008), 2334-2345.
\bibitem{M08} S. Mallat, \emph{A Wavelet Tour of Signal Processing: The Sparse Way}, Academic Press, 2008.
\bibitem{MPT08} S. Mendelson, A. Pajor and N. Tomczak-Jaegermann, \emph{Uniform uncertainty
principle for Bernoulli and subgaussian ensembles,} Constr. Approx.,
\textbf{28} (2008), 277-289.
\bibitem{NDEG13}A. S. Nam, M. E. Davies, M. Elad and R. Gribonval, \emph{The
cosparse analysis model and algorithms}, Appl. Comput. Harmon.
Anal., \textbf{34} (2013), 30-56.

\bibitem{RV08} M. Rudelson and R. Vershynin, \emph{On sparse reconstruction from
Fourier and Gaussian measurements}, Comm. Pure Appl. Math.,
\textbf{61} (2008), 1025-1045.
\bibitem{SCY08} R. Saab, R. Chartrand and O. Yilmaz, \emph{Stable sparse approximations via nonconvex
optimization}, Int. Conf. Acoust. Spee., (2008), 3885-3888.
\bibitem{SL12}  Y. Shen and S. Li, \emph{Restricted $p$-isometry property and its application for nonconvex
compressive sensing}, Adv. Comput. Math., \textbf{37} (2012),
441-452.
\bibitem{S11} Q. Sun, \emph{Sparse approximation property and stable recovery of sparse signals from noisy measurements},
IEEE Trans. Signal Processing, \textbf{19} (2011), 5086-5090.
\bibitem{TEBN14} Z. Tan, Y. C. Eldar, A. Beck and A. Nehorai, \emph{Smoothing and decomposition for analysis sparse recovery}, IEEE Trans. Signal Processing, \textbf{62} (2014), 1762-1774.
\bibitem{ZYY14}H. Zhang, M. Yan, W. Yin, \emph{One condition for solution uniqueness and robustness of both $l_1$-synthesis and $l_1$-analysis minimizations}, Arxiv, 2013.
\end{thebibliography}
\end{document}